\documentclass[12pt]{article}
\usepackage{amsmath, amsthm, amsfonts}
\usepackage[margin=1 in]{geometry}
\usepackage{blindtext, xcolor}
\usepackage{algorithm}
\usepackage{algorithmicx}
\usepackage{algpseudocode}
\usepackage{apacite}
\usepackage{natbib}
\usepackage{appendix}
\usepackage{rotating}
\usepackage{hyperref}
\usepackage{float}
\usepackage{graphicx}
\usepackage{caption,subcaption}
\usepackage{cleveref}
\usepackage{array}
\usepackage{arydshln}
\usepackage{authblk}
\usepackage{xr,bm}
\usepackage{mathrsfs}
\usepackage{multirow}
\usepackage{diagbox}
\usepackage{makecell}
\usepackage{enumitem}
\usepackage{breqn}
\usepackage[labelfont=bf]{caption}

\newcommand{\be} {\begin{eqnarray*}}
\newcommand{\ee} {\end{eqnarray*}}

\theoremstyle{definition}

\newtheorem{theorem}{Theorem}[section]
\newtheorem*{theorem*}{Theorem}

\newtheorem{proposition}[theorem]{Proposition}

\def\*#1{\bm{#1}}

\title{Two-Sample Hypothesis Testing for Large Random Graphs of Unequal Size}
\author[1]{Xin Jin}
\author[1]{Kit Chan}
\author[2]{Ian Barnett}
\author[1]{Riddhi Pratim Ghosh}
\affil[1]{Department of Mathematics and Statistics, Bowling Green State University}
\affil[2]{Department of Biostatistics, University of Pennsylvania}
\date{}

\begin{document}
\maketitle
\begin{abstract}	
	Two-sample hypothesis testing for large graphs is popular in cognitive science, probabilistic machine learning and artificial intelligence. While numerous methods have been proposed in the literature to address this problem, less attention has been devoted to scenarios involving graphs of unequal size or situations where there are only one or a few samples of graphs. In this article, we propose a Frobenius test statistic tailored for small sample sizes and unequal-sized random graphs to test whether they are generated from the same model or not. Our approach involves an algorithm for generating bootstrapped adjacency matrices from estimated community-wise edge probability matrices, forming the basis of the Frobenius test statistic. We derive the asymptotic distribution of the proposed test statistic and validate its stability and efficiency in detecting minor differences in underlying models through simulations. Furthermore, we explore its application to fMRI data where we are able to distinguish brain activity patterns when subjects are exposed to sentences and pictures for two different stimuli and the control group. 	
\end{abstract}
%%%%%%%%%%%%%%%%%%%%%%%%%%%%%%%%%%%%%%%%%%%%%%%%%%%%%%
%%%%%%%%%%%%%%%%%%%%%%%%%%%%%%%%%%%%%%%%%%%%%%%%%%%%%%
\vspace{1cm}


\newpage
\section{Introduction} \label{intro}

Network data, consisting of a set of nodes and edges, can be broadly used to represent interactions among a group of agents or entities, and one can find its applications across almost every scientific field. Examples of such study abound in neuroscience \citep{bassett2008hierarchical}, gene regulatory networks \citep{zhang2009differential}, ecological networks \citep{montoya2006ecological}, and social interactions \citep{kossinets2006empirical,eagle2009inferring}, etc. Within network data, community structure pertains to identifying and characterizing groups or clusters of nodes that exhibit a higher degree of connectivity among themselves compared to the rest of the network \citep{newman2004finding, newman2006modularity}. For instance, in biological networks, communities could correspond to groups of interacting proteins or genes with related functions. Recent studies have made significant developments in community detection, as evident in works by \cite{bickel2016hypothesis}, \cite{ma2021asymptotic}, \cite{yanchenko2021generalized}, \cite{jha2022multiple}, and \cite{ghosh2023selecting}.

Identifying structural patterns within a network facilitates statistical inferences or testing hypotheses about the network. The latter has drawn considerable attention recently, especially in various applications. For instance, testing for similarity across brain graphs is an area of active research at the intersection of neuroscience and machine learning, and practitioners often utilize classical parametric two-sample tests, such as edgewise t-tests on correlations. Our objective in this article is to establish a concrete framework for a specific two-sample graph testing problem and to introduce a robust and computationally feasible test statistic for this purpose. In the realm of hypothesis testing in functional neuroimaging network data, \cite{ginestet2017hypothesis} emphasizes the need for a global modeling strategy to compare groups of subject-specific networks and highlights the mathematical behavior of networks without structural constraints. \cite{gao2018goodness} introduces a testing procedure based on indefinite integrals of symmetric polynomials, specifically designed for i.i.d. samples from a categorical distribution. In addition, \cite{kim2021two} presents innovative methods for comparing high-dimensional Markov networks without the need for estimating individual graphs. The widespread application of network comparison has made two-sample hypothesis testing a focal point of significant interest. The work by \cite{tang2017nonparametric} investigates testing whether two random dot product graphs share generating latent positions from the same distribution. \cite{ghoshdastidar2017two} proposes a novel approach to two-sample hypothesis testing with only one observation from each model, emphasizing the concentration of network statistics and providing a consistent and minimax optimal test, particularly applicable to comparing large random graphs like social networks on platforms such as Facebook and LinkedIn. While in \cite{ghoshdastidar2018practical}, two-sample testing of large graphs is further extended to small population set-up, assuming the same set of vertices. The challenges of testing between two groups of sparse graphs are addressed in the work by \cite{ghoshdastidar2020two}, particularly when the sample size is significantly smaller than the number of vertices. \cite{xia2022hypothesis} examines the comparison of two population means of network data, focusing on individual network links and using symmetric matrices. \cite{chen2023hypothesis} proposes a general procedure for hypothesis testing of networks and differentiating distributions of two samples of networks. \cite{tang2017semiparametric} proposed a test statistic using the function of a spectral decomposition of the adjacency matrix for testing two random dot product graphs limited to a common vertex set.

Despite the prevailing focus on a shared set of vertices in earlier studies, \cite{ghoshdastidar2017two} stands out by addressing two-sample random graphs of varying sizes. Nevertheless, the applicability of their method is limited to scenarios with only one observation from each model. Our work is motivated by the pursuit of a new procedure and test statistic that can be employed for two-sample random graphs of different sizes, free from such restrictions. In this article, we extend the discourse to address the challenges posed by unequal-sized random graphs, offering a robust methodology capable of discerning subtle distinctions between network structures.

We begin with a mild assumption that the number of communities in two random graphs is equal. If two random graphs have different numbers of communities, it is evident that they are not generated from the same underlying model. This viewpoint often lacks consideration and discussion in many two-sample hypothesis testing studies, see \cite{ghoshdastidar2017two}, \cite{ghoshdastidar2018practical}, and \cite{ghoshdastidar2020two}. We highlight it here because our innovative method effectively utilizes this assumption. To validate this assumption and estimate the number of communities, various methods can be employed. Likelihood-based approaches, as demonstrated by \cite{ma2021determining}, utilize the likelihood function or an approximate pseudo-likelihood function. These approaches are effective for selecting the best model and estimating the number of communities in both dense and sparse degree-corrected stochastic block models. Cross-validation-based methodologies, illustrated by \cite{chen2018network} and \cite{li2020network}, involve generating multiple network copies through network resampling techniques. Subsequently, cross-validation is employed to determine the optimal number of communities. On the other hand, spectral methods \citep{lei2016goodness, li2020network, jha2022multiple} leverage the spectral features of suitably adjusted adjacency matrices to estimate the number of communities. 

After estimating the number of communities, the original adjacency matrices are transformed into community-wise edge probability matrices.
Within the stochastic block model framework, bootstrapped adjacency matrices are then generated based on the estimated community-wise edge probability matrices. A Frobenius test statistic is built to quantify the squared element-wise differences between two sequences of bootstrapped adjacency matrices, effectively addressing the issue of unequal graph sizes. As the number of bootstraps increases, the proposed test statistic converges towards a normal distribution, with mean and variance characterized by factors including the sample size, the size of the bootstrapped adjacency matrices, and the number of communities. In the application to fMRI data, a sequential procedure is employed to determine the number of communities. The Frobenius test statistic adequately discerns differences in brain activity between exposure to sentences and pictures.

The rest of this article is organized as follows: Section~\ref{prelim_sec} introduces the terminologies and notations used. In Section~\ref{method_sec}, the algorithm for generating bootstrapped adjacency matrices and the definition of Frobenius test statistic are discussed, along with simulation where the proposed test statistic is then evaluated, demonstrating superior test power compared to the method in \cite{ghoshdastidar2018practical}. Section~\ref{app_sec} applies the method to fMRI data, validating the stability and efficiency of the Frobenius test statistic. Final remarks are provided in Section~\ref{dis_sec}.

%%%%%%%%%%%%%%%%%%%%%%%%%%%%%%%%%%%%%%%%%%%%%%%%%%%%%%%%%%%%%%%%%%%%%%%%%%

\section{Preliminaries} \label{prelim_sec}
In this section, we introduce terminologies and notations that are used throughout the
paper. We study two sequences of networks that each consist of $m$ random graphs, i.e., $(G_l)_{l \in \{1,\ldots,m \}}$ and $(H_l)_{l \in \{1,\ldots,m \}}$, where $m$ represents the sample size and is relatively small. The number of vertices in a graph is denoted by $n$, which determines graph size.

To be more flexible, the following notations are introduced for individual graphs $G$ and $H$. Consider $K$ communities in both graphs $G$ and $H$, with the total number of nodes being $n^G$ and $n^H$, respectively. If $n_u^G$ denotes the size of the $u \textsuperscript{th}$ community in graph $G$ and $n_v^H$ denotes the size of the $v \textsuperscript{th}$ community in graph $H$, then 
\begin{center}
	$n^G = \displaystyle \sum_{u=1}^K n_u^G$ and $n^H = \displaystyle \sum_{v=1}^K n_v^H$.
\end{center}

The Erd\H{o}s-R\'enyi model, proposed by mathematicians Paul Erd\H{o}s and Alfr\'ed R\'enyi in the late 1950s and early 1960s, is a foundational concept in the field of random graph theory. This model provides a framework for studying the properties of random graphs and has been influential in understanding the emergence of complex structures in various systems. The inhomogeneous Erd\H{o}s-R\'enyi model is a random graph model that generalizes the classic Erd\H{o}s-R\'enyi model by allowing for heterogeneity in the edge probabilities, see \cite{bollobas2007phase}. In this model, a graph with $n$ vertices is generated by specifying a probability distribution on the edge probabilities, which may depend on various factors such as the degrees of the vertices or their location in the network. Specifically, the inhomogeneous Erd\H{o}s-R\'enyi model generates a random graph $G(n, \bm P)$, where the probability of an edge between vertices $i$ and $j$ is given by $P_{ij}$, which may be a function of various vertex attributes. The model allows for a wide range of edge probabilities, leading to a more diverse network structure than the classic Erd\H{o}s-R\'enyi model, which assumes a uniform edge probability across all pairs of vertices.

We assume that the graphs are generated from the inhomogeneous Erd\H{o}s-R\'enyi (IER) model. Suppose that graphs $(G_l)_{l \in \{1,\ldots,m \}} \overset{\text{iid}}{\sim} \text{IER}(\bm P)$ and $(H_l)_{l \in \{1,\ldots,m \}} \overset{\text{iid}}{\sim} \text{IER}(\bm Q)$, where $\bm P \in [0, 1]^{n^G \times n^G}$ and $\bm Q \in [0, 1]^{n^H \times n^H}$ are edge probability matrices. Then the adjacency matrices $\bm A^G \in \{0, 1\}^{n^G \times n^G}$ and $\bm A^H \in \{0, 1\}^{n^H \times n^H}$ satisfy:
\begin{enumerate}[label=(\roman*)]
	\item $(\bm A^G)_{ij} \sim \text{Bernoulli}(P_{ij})$ for all $i, j \in \{1, \ldots, n^G \}$. \item ${(\bm A^G)_{ij}}$ are mutually independent for all $i, j \in \{1, \ldots, n^G \}$.
	\item $(\bm A^H)_{ij} \sim \text{Bernoulli}(Q_{ij})$ for all $i, j \in \{1, \ldots, n^H \}$. \item ${(\bm A^H)_{ij}}$ are mutually independent for all $i, j \in \{1, \ldots, n^H \}$.
\end{enumerate}

$\widehat{\bm P}_{\text{com}}$ and $\widehat{\bm Q}_{\text{com}}$ represent estimated community-wise edge probability matrices, where each entry denotes the probability of edge presence within and between communities. $\widetilde{\bm A}$ denotes the bootstrapped adjacency matrix generated from $\widehat{\bm P}_{\text{com}}$ and $\widehat{\bm Q}_{\text{com}}$ through sampling from the stochastic block model of networks. For example, $\widetilde{\bm A}^G_w$ corresponds to the $w\textsuperscript{th}$ bootstrapping sample generated from $\widehat{\bm P}_{\text{com}}$ for graph $G$. In this article, we consistently use $X_1, X_2, \ldots$ and $Y_1, Y_2, \ldots$ to represent random variables.

%%%%%%%%%%%%%%%%%%%%%%%%%%%%%%%%%%%%%%%%%%%%%%%%%%%%%%%%%%%%%%%%%%%%%%%%%%

\section{Methodology} \label{method_sec}
Recent studies addressing two-sample random graph testing problems have mainly focused on equal-sized graphs. However, in reality, many unresolved detection issues persist with unequal random graphs that await exploration. Our goal is to bridge this gap. In this section, we introduce an innovative approach to address this issue, utilizing both concentration-based and bootstrapping techniques. Rather than directly constructing the test statistic from adjacency matrices, we extract approximate community-wise edge probability matrices. Thereafter, we formulate a Frobenius test statistic, which is capable of detecting both directed and undirected graphs, based on these matrices.

%%%%%%%%%%%%%%%%%%%%%%%%%%%%%%%%%%%%%%%%%%%%%%%%%%%%%%%%%%%%%%%%%%%%%%%%%%%%%%%%%%%%%%%%

\subsection{Bootstrapping adjacency matrices}
When the sample size $m$ is small, it is of interest to test whether two sequences of random graphs are from the same population based on asymptotics for undirected and unweighted large graphs, namely, test whether $(G_l)_{l \in \{1,\ldots,m\}}$ and $(H_l)_{l \in \{1,\ldots,m\}}$ are generated from the same random model.

Most recent studies \citep{ghoshdastidar2018practical, ghoshdastidar2020two} have focused on testing two graphs on a common vertex set of size $n$. For any $n$, suppose that $(G_l)_{l \in \{1,\ldots,m\}} \overset{\text{iid}}{\sim} \text{IER}(\bm P)$ and $(H_l)_{l \in \{1,\ldots,m\}} \overset{\text{iid}}{\sim} \text{IER}(\bm Q)$, where $\bm P$ and $\bm Q \in [0, 1]^{n \times n}$ are edge probability matrices, test the hypothesis:
\begin{center}
	$H_0: \bm P = \bm Q$ vs. $H_a: \bm P \neq \bm Q$.
\end{center}
The use of graphs that possess equal size for testing purposes, while beneficial, may present constraints that limit their overall utility. Hence, we have been driven to investigate alternative testing approaches that do not rely on the equal-size assumption, thereby expanding the scope of our testing methods.

Community detection has been broadly studied in network analysis, if two random graphs have different numbers of communities, then there is no need for the test described above. Naturally, we assume that two sequences of random graphs $(G_l)_{l \in \{1,\ldots,m\}}$ and $(H_l)_{l \in \{1,\ldots,m\}}$ have the same number of communities. Let $C_u$ be the $u \textsuperscript{th}$ community of $K$-community graphs, $u \in \{1, \ldots, K\}$. Then community-wise edge probability matrices $\bm P_{\text{com}}$ and $\bm Q_{\text{com}} \in [0, 1]^{K \times K}$ are defined as follows,
\begin{align} \label{P&Q}
	\bm P_{\text{com}} = \begin{pmatrix}
		p_{11} & q_{12} & \cdots & q_{1K} \\
		q_{21} & p_{22} & \cdots & q_{2K} \\
		\vdots  & \vdots  & \ddots & \vdots  \\
		q_{K1} & q_{K2} & \cdots & p_{KK}
	\end{pmatrix},
	& \text{ and }
	\bm Q_{\text{com}} = \begin{pmatrix}
		p_{11}^* & q_{12}^* & \cdots & q_{1K}^* \\
		q_{21}^* & p_{22}^* & \cdots & q_{2K}^* \\
		\vdots  & \vdots  & \ddots & \vdots  \\
		q_{K1}^* & q_{K2}^* & \cdots & p_{KK}^*
	\end{pmatrix},
\end{align}
where $p_{uu}$ is $C_u$ within-community probability and $q_{uv}$ is the probability between $C_u$ and $C_v$ of graph $G$, $u, v \in \{1, \ldots, K\}$, similarly for entries in $\bm Q_{\text{com}}$ of graph $H$. We test the hypothesis
\begin{center}
	$H_0: \bm P_{\text{com}} = \bm Q_{\text{com}}$ vs. $H_a: \bm P_{\text{com}} \neq \bm Q_{\text{com}}$.
\end{center}

It's worth noting that the hypothesis above is the most similar to that presented in \cite{ghoshdastidar2018practical}, while it specifically addresses graphs of unequal sizes within the framework of IER models. To assess the aforementioned hypothesis, we initiate the process by generating two sets of random graph sequences derived from stochastic block models. These models incorporate true community-wise edge probability matrices denoted as $\bm P_{\text{com}}$ and $\bm Q_{\text{com}}$. For the sake of simplicity, we will consider an individual graph $G$, and its adjacency matrix $\bm A^G$ can be partitioned into $K^2$ blocks along with the $K$ communities. Subsequently, the matrices $\bm A^G$ and the estimated community-wise edge probability matrix $\widehat{\bm P}_{\text{com}} \in [0, 1]^{K \times K}$ can be constructed in the following manner:
\begin{center}
	$\bm A^G = \begin{pmatrix}
		\bm A_{11} & \bm A_{12} & \cdots & \bm A_{1K} \\
		\bm A_{21} & \bm A_{22} & \cdots & \bm A_{2K} \\
		\vdots  & \vdots  & \ddots & \vdots  \\
		\bm A_{K1} & \bm A_{K2} & \cdots & \bm A_{KK}
	\end{pmatrix}$,
	$\widehat{\bm P}_{\text{com}} = \begin{pmatrix}
		\hat{p}_{11} & \hat{q}_{12} & \cdots & \hat{q}_{1K} \\
		\hat{q}_{21} & \hat{p}_{22} & \cdots & \hat{q}_{2K} \\
		\vdots  & \vdots  & \ddots & \vdots  \\
		\hat{q}_{K1} & \hat{q}_{K2} & \cdots & \hat{p}_{KK}
	\end{pmatrix}$,
\end{center}
where for all $u, v \in \{1, \ldots, K\}$,
\begin{center}
	$\displaystyle \hat{p}_{uu} = \frac{\# \text{ of } 1\text{'s in } \bm A_{uu}}{2{n_u^G \choose 2}}$, and $\displaystyle \hat{q}_{uv} = \frac{\# \text{ of } 1\text{'s in } \bm A_{uv}}{{n_u^G \choose 1}{n_v^G \choose 1}}$ for $u \neq v$. 
\end{center}

Similarly, the estimated community-wise edge probability matrix $\widehat{\bm Q}_{\text{com}} \in [0, 1]^{K \times K}$ of graph $H$ is denoted by 
\begin{center}
	$\widehat{\bm Q}_{\text{com}} = \begin{pmatrix}
		\hat{p}_{11}^* & \hat{q}_{12}^* & \cdots & \hat{q}_{1K}^* \\
		\hat{q}_{21}^* & \hat{p}_{22}^* & \cdots & \hat{q}_{2K}^* \\
		\vdots  & \vdots  & \ddots & \vdots  \\
		\hat{q}_{K1}^* & \hat{q}_{K2}^* & \cdots & \hat{p}_{KK}^*
	\end{pmatrix}$.
\end{center}
The unbiased property of estimates $\hat{\bm p}, \hat{\bm p}^*, \hat{\bm q}$, $\hat{\bm q}^*$ is provided in \Cref{A.4.1}.

The size of the bootstrapped adjacency matrices, denoted as $n$, is selected as $n = \max \left(iK: i \in \mathbb{Z}, n \leq \min(n^G, n^H) \right)$. The bootstrapped adjacency matrices $\widetilde{\bm A}^G_w \in \{0, 1\}^{n \times n}$ for $w \in \{1, \ldots, d\}$, can be generated from $\widehat{\bm P}_{\text{com}}$ through the stochastic block model with $d$ goes to infinity. Similarly, the adjacency matrices $\widetilde{\bm A}^H_w \in \{0, 1\}^{n \times n}$ for $w \in \{1, \ldots, d\}$ can be generated from $\widehat{\bm Q}_{\text{com}}$. In this bootstrapping process, the vertex set is partitioned into $K$ communities in both a balanced and imbalanced manner. The number of nodes in the $u\textsuperscript{th}$ community is denoted by $n_u$, which is calculated as $n \exp{(w_u)}/\sum_{u=1}^K \exp{(w_u)}$, where $w_u \sim \mathcal{N}(0, \tau^2)$, that controls the structure of the graph. For example, $\tau = 0$, and $\tau > 0$ etc.

We propose the following procedure to generate bootstrapped adjacency matrices. 
\begin{table}[H]
	\captionof{table}{\bf{Algorithm for generating bootstrapped adjacency matrices}} \label{tab1}
	\centering 	
	\setlength{\tabcolsep}{0.9em}
	\renewcommand{\arraystretch}{1.1}
	
	\smallskip
	\small
	\begin{tabular}{l}
		\hline
		\hline
		\rule{0pt}{0.6cm}
		\textbf{Input:} The original adjacency matrices $(\bm A^{G_l})_{l \in \{1,\ldots,m\}} \in \{0, 1\}^{n^G \times n^G}$ and $(\bm A^{H_l})_{l \in \{1,\ldots,m\}}$ \\
		\hspace{1.5cm} $\in \{0, 1\}^{n^H \times n^H}$.
		\vspace{5pt}
		\\ \hline
		\rule{0pt}{0.6cm}
		1. Estimate community-wise edge probability matrices $(\widehat{\bm P}^{G_l}_{\text{com}})_{l \in \{1, \ldots, m\}}$ and $(\widehat{\bm Q}^{H_l}_{\text{com}})_{l \in \{1, \ldots, m\}}$ \\
		\hspace{0.5cm} $\in [0, 1]^{K \times K}$. \\ 
		\vspace{5pt}
		2. Choose $n = \max \left(iK: i \in \mathbb{Z}, n \leq \min(n^G, n^H) \right)$. \\
		\vspace{5pt}
		3. Determine the block size: For each community $u \in \{1, \ldots, K\}$, $n_u = n \exp{(w_u)}/\sum_{u=1}^K \exp{(w_u)}$, \\  
		\vspace{5pt}
		\hspace{0.4cm} where $w_u \sim \mathcal{N}(0, \tau^2)$. \\
		\hspace{0.4cm} (i) Balanced block sizes:  $\tau = 0$. \\
		\hspace{0.4cm} (ii) Imbalanced block sizes: $\tau > 0$. \\
		\vspace{5pt}
		4. Generate bootstrapped adjacency matrices $\widetilde{\bm A}^{G_l}_w \text{ and } \widetilde{\bm A}^{H_l}_w \in \{0, 1\}^{n \times n}$ from $(\widehat{\bm P}^{G_l}_{\text{com}})_{l \in \{1, \ldots, m\}}$ \\ 
		\hspace{0.4cm} and $(\widehat{\bm Q}^{H_l}_{\text{com}})_{l \in \{1, \ldots, m\}}$ through the stochastic block model for $w \in \{1, \ldots, d\}$.
		\vspace{5pt}
		\\ \hline
		\rule{0pt}{0.6cm}
		\textbf{Output:} The bootstrapped adjacency matrices $(\widetilde{\bm A}^{G_l}_w)_{l \in \{1, \ldots, m\}} \text{ and } (\widetilde{\bm A}^{H_l}_w)_{l \in \{1, \ldots, m\}}$ for $w$ \\
		\vspace{5pt}
		\hspace{1.7cm} $\in \{1, \ldots, d\}$. \\ 
		\hline
		\hline
	\end{tabular}
\end{table}

%%%%%%%%%%%%%%%%%%%%%%%%%%%%%%%%%%%%%%%%%%%%%%%%%%%%%%%%%%%%%%%%%%%%%%%%%%

\subsection{Frobenius Statistic} 
As introduced at the beginning of \Cref{prelim_sec}, our goal is to assess if the underlying probabilistic models of two random graphs are the same or not. To achieve this, establishing the correspondence between communities in matrices $\bm{P}_{\text{com}}$ and $\bm{Q}_{\text{com}}$ is crucial. This involves examining all possible permutations $\pi$ of communities in one matrix, e.g., $\bm{Q}_{\text{com}}$. This exploration allows us to define a metric, denoted as $T$, measuring dissimilarity between bootstrapped adjacency matrices $(\widetilde{\bm{A}}^{G_l}_w)_{l \in \{1, \ldots, m\}}$ and $(\widetilde{\bm{A}}^{H_l}_w)_{l \in \{1, \ldots, m\}}$ as outlined below.

Given two sequences of bootstrapped adjacency matrices $(\widetilde{\bm A}^{G_l}_w)_{l \in \{1, \ldots, m\}}$ and $(\widetilde{\bm A}^{H_l}_w)_{l \in \{1, \ldots, m\}}$, where $w \in \{1, \ldots, d\}$, consider the test statistic based on estimates of the Frobenius norm 
\begin{center}
	$T = \displaystyle \inf_{\pi}
	\left( \frac{\displaystyle 
		\sum_{w=1}^d
		\sum_{l=1}^m 
		\sum_{i, j=1}^n
		\left((\widetilde{\bm A}^{G_l}_w)_{ij} - (\widetilde{\bm A}^{H_l}_w)_{\pi (\widehat{\bm Q}_{\text{com}})(ij)} \right)^2}{mdn(n-1)}
	\right)$,
\end{center}
where the infimum is taken over all possible permutations $\pi$ of $K$ communities in graph $H$, and $(\widetilde{\bm A}^{G_l}_w)_{ij}$ is the $(i, j)\textsuperscript{th}$ entry of $\widetilde{\bm A}^{G_l}_w$ and $(\widetilde{\bm A}^{H_l}_w)_{\pi (\widehat{\bm Q}_{\text{com}}) (ij)}$ is the $(i, j) \textsuperscript{th}$ entry of $\widetilde{\bm A}^{H_l}_w$ after applying $\pi$ to the matrix $\widehat{\bm Q}_{\text{com}}$.

Under the null hypothesis, for any fixed $i, j$ with $1\leq i, j  \leq n$, the random variables $(\widetilde{\bm A}_w^{G_l})_{ij}$ and 
$(\widetilde{\bm A}_w^{H_l})_{\pi (\widehat{\bm Q}_{\text{com}})(ij)}$, where $w \in \{1, \ldots, d\}$ and $l \in \{1, \dots, m\}$, are i.i.d. with the same mean $\mu_{ij}$ and variance $\sigma^2_{ij}$. Next, to analyze the asymptotic distribution of the Frobenius statistic denoted as $T$, we establish the following analogue of the Lindeberg condition for the difference $\widetilde{\bm A}_w^{G_l} - (\widetilde{\bm A}_w^{H_l})_{\pi (\widehat{\bm Q}_{\text{com}})}$ from two comparable sets of graphs $G_l$ and $H_l$.

\begin{proposition} \label{prop3.1}
	Under the null hypothesis, there exists a permutation $\pi$ of the matrix $\widehat{\bm Q}_{\text{com}}$ such that for any given $\delta >0$, and integers $n, m \geq 1$, we have the following limit as $d \rightarrow \infty$.
	\begin{align*}
		\begin{split}
			\left( \frac{1}{mds_1^2} \right) \sum_{w=1}^d
			\sum_{l=1}^m 
			\sum_{i, j=1}^n
			E \left[ \left((\widetilde{\bm A}_w^{G_l})_{ij} - (\widetilde{\bm A}_w^{H_l})_{\pi (\widehat{\bm Q}_{\text{com}}) (ij)} \right)^2 \,1_{\left\{\left|(\widetilde{\bm A}_w^{G_l})_{ij} - (\widetilde{\bm A}_w^{H_l})_{\pi (\widehat{\bm Q}_{\text{com}}) (ij)}\right|\ > \ \delta mds_1 \right\}} \right] \rightarrow 0,
		\end{split}
	\end{align*} 
	where $s_1^2 = 2 \displaystyle        \sum_{i, j=1}^n \sigma_{ij}^2$.
\end{proposition} 

The proof of \Cref{prop3.1} is deferred to \Cref{appendix_a1}. According to \Cref{prop3.1}, the Lindeberg condition is satisfied. As $d \rightarrow \infty$, the central limit theorem implies that $T$ converges to a normal distribution, referred to \cite{Lindeberg1922}. 

We are now introducing the asymptotic behavior of $T$ as $d$ approaches infinity. The parameters of the asymptotic normal distribution under $H_0$ are further explored in both balanced and imbalanced block size scenarios. For all $w \in \{1, \ldots, d\}$, 
\begin{enumerate}[label=(\roman*)]
	\item if $i, j \in C_u$, then $\left((\widetilde{\bm A}_w^{G})_{ij} - (\widetilde{\bm A}_w^{H})_{\pi(\widehat{\bm Q}_{\text{com}})(ij)}\right)^2 \sim \text{Bernoulli}(\hat{p}_{uu} + \hat{p}^*_{uu} - 2\hat{p}_{uu} \hat{p}^*_{uu})$.
	
	\item if $i \in C_u$ and $j \in C_v$ with $u \neq v$, then $\left((\widetilde{\bm A}_w^{G})_{ij} - (\widetilde{\bm A}_w^{H})_{\pi(\widehat{\bm Q}_{\text{com}})(ij)}\right)^2 \sim \text{Bernoulli}(\hat{q}_{uv} + \hat{q}^*_{uv} - 2\hat{q}_{uv}\hat{q}^*_{uv})$. 
\end{enumerate}
Denote $\tilde{p}_{uu} = \hat{p}_{uu} + \hat{p}^*_{uu} - 2\hat{p}_{uu}\hat{p}^*_{uu}$ and $\tilde{q}_{uv} = \hat{q}_{uv} + \hat{q}^*_{uv} - 2\hat{q}_{uv}\hat{q}^*_{uv}$. 

\begin{theorem}\label{theorem3.2}
	Under the null hypothesis, the Frobenius statistic $T \sim \mathcal{N}(\mu, \sigma^2)$ as $d \rightarrow \infty$. The specific expressions for $\mu$ and $\sigma^2$ are presented in two cases as follows.
	\newline
	$A1.$ For the balanced block size, 
	\begin{center}
		$\mu = X_1 + X_2$ and $\sigma^2 = X_3 + X_4$, 
	\end{center}
	where 
	\begin{center}
		$\displaystyle X_1 = \frac{n-K}{m(n-1)K^2} \sum_{l=1}^m \sum_{u=1}^{K} \tilde{p}^l_{uu}$, 
		$\displaystyle X_2 = \frac{n}{m(n-1)K^2} \sum_{l=1}^m \sum_{u=1}^{K} \sum_{v \neq u} \tilde{q}^l_{uv}$, \\
		$\displaystyle X_3 = \frac{n-K}{m^2dn(n-1)^2K^2} \sum_{l=1}^m \sum_{u=1}^{K}\tilde{p}^l_{uu}(1 - \tilde{p}^l_{uu})$, and
		$\displaystyle X_4 = \frac{n}{m^2dn(n-1)^2K^2} \sum_{l=1}^m \sum_{u=1}^{K} \sum_{v \neq u} \tilde{q}^l_{uv}(1 - \tilde{q}^l_{uv})$. 
	\end{center} 
	\leavevmode
	\newline
	$A2.$ For the imbalanced block size, 
	\begin{center}
		$\mu = Y_1 + Y_2$ and $\sigma^2 = Y_3 + Y_4$, 
	\end{center}
	where 
	\begin{center}
		$Y_1 = \frac{\displaystyle \sum_{l=1}^m \sum_{u=1}^K \tilde{p}^l_{uu} \left[n \exp^2(w_u) -  \exp(w_u)\sum_{u=1}^K \exp(w_u) \right]}{\displaystyle m(n-1) \left(\sum_{u=1}^K \exp(w_u)\right)^2}$, 
		$Y_2 = \frac{\displaystyle n \sum_{l=1}^m \sum_{u=1}^K \sum_{v \neq u} \tilde{q}^l_{uv} \left[\exp(w_u) \exp(w_v) \right]}{\displaystyle m(n - 1) \left(\sum_{u=1}^K \exp(w_u)\right)^2}$,\\
		$Y_3 = \frac{\displaystyle \sum_{l=1}^m \sum_{u=1}^K \tilde{p}^l_{uu} (1 - \tilde{p}^l_{uu}) \left[n \exp^2(w_u) -  \exp(w_u)\sum_{u=1}^K \exp(w_u) \right]}{\displaystyle m^2dn(n - 1)^2 \left(\sum_{u=1}^K \exp(w_u)\right)^2}$, and\\
		$Y_4 = \frac{\displaystyle \sum_{l=1}^m \sum_{u=1}^K \sum_{v \neq u} \tilde{q}^l_{uv} (1 - \tilde{q}^l_{uv}) \left[\exp(w_u) \exp(w_v) \right]}{\displaystyle m^2d(n - 1)^2 \left(\sum_{u=1}^K \exp(w_u)\right)^2}$.
	\end{center}
\end{theorem}

For the following proposition, we assume $\gamma$ is positive because $n$ is a fixed and relatively small number comparing to $d$. So if we take the condition $s_2^2 \geq cd^{-2 \gamma}$, which means there exists a positive constant $c$ such that $n(n-1)/2 \geq cd^{-2 \gamma}$, or $d^{2 \gamma} \geq 2c/n(n-1)$. Hence, we can assume $0 < \gamma < 1$. In Proposition \ref{prop3.3}, conditions $B1$ and $B2$ establish a tight lower bound on the variance in bootstrapped adjacency matrices, dependent on network size $n$ and community size $K$. This ensures that the network density does not become excessively sparse or dense.

\begin{proposition} \label{prop3.3}
	Under the alternative hypothesis, if there exists a value $\gamma$ satisfying $0 < \gamma < 1$ such that
	\newline
	$B1.$ for balanced block sizes, 
	\begin{align} \label{cond_1}
		\begin{split}
			\displaystyle \frac{K^3}{n(n-K)} + \frac{K^4 - K^3}{n^2} = \Omega(d^{- 2\gamma}).   
		\end{split}
	\end{align} 
	\newline
	$B2.$ for imbalanced block sizes,
	\begin{align} \label{cond_2}
		\begin{split}
			\displaystyle \sum_{u=1}^K \frac{\left(\sum_{u=1}^K \exp{(w_u)} \right)^2}{n^2 \left(\exp{(w_u)} \right)^2 - n\exp{(w_u)} \displaystyle \sum_{u=1}^K \exp{(w_u)}} + 
			\displaystyle \sum_{u=1}^K \sum_{v \neq u} \frac{\left(\sum_{u=1}^K \exp{(w_u)} \right)^2}{n^2 \exp{(w_u)} \exp{(w_v)}} = \Omega(d^{- 2\gamma}).
		\end{split}
	\end{align} 
	Then there exists a permutation $\pi$ of the matrix $\widehat{\bm Q}_{\text{com}}$ such that for any given $\delta > 0$, and integers $n, m \geq 1$, we have the following limit as $d \rightarrow \infty$.
	\begin{align*}
		\begin{split}
			\left(\frac{1}{mds_2^2 } \right) \sum_{w=1}^d
			\sum_{l=1}^m 
			\sum_{i, j=1}^n
			E \left[ \left((\widetilde{\bm A}_w^{G_l})_{ij} - (\widetilde{\bm A}_w^{H_l})_{\pi (\widehat{\bm Q}_{\text{com}})(ij)} - \mu \right)^2 \,1_{\left \{\left|(\widetilde{\bm A}_w^{G_l})_{ij} - (\widetilde{\bm A}_w^{H_l})_{\pi (\widehat{\bm Q}_{\text{com}})(ij)} - \mu \right|\ > \ \delta mds_2 \right \}} \right] \rightarrow 0,
		\end{split}
	\end{align*} 
	where $\mu = \mu_{ij}^G - \mu_{ij}^H$ and $s_2^2 = \displaystyle \sum_{i, j=1}^n\, \left({\sigma_{ij}^2}^G + {\sigma_{ij}^2}^H \right)$.
\end{proposition}

\begin{proof}[\textbf{\upshape Proof of Proposition 3.3:}]
	Suppose that \eqref{cond_1} and \eqref{cond_2} are satisfied, then $s_2^2 \geq cd^{-2 \gamma}$ for a positive constant $c$, refer to \Cref{A.4.2} for more details. Notice that $c$ is a function of the true parameters $\bm p, \bm p^*, \bm q$ and $\bm q^*$, which are vectors of edge probabilities in \eqref{P&Q}. For large enough $d$, $\delta mds_2 \geq \sqrt{c} \delta md^{1-\gamma}$, which goes to $\infty$, as $d \rightarrow \infty$. Note that random variable $(\widetilde{\bm A}_w^{G_l})_{ij} - (\widetilde{\bm A}_w^{H_l})_{\pi (\widehat{\bm Q}_{\text{com}})(ij)}$ takes values of -1, 0, 1, and also that $\mu = \mu_{ij}^G - \mu_{ij}^H \in [-1, 1]$. The value of $\left|(\widetilde{\bm A}_w^{G_l})_{ij} - (\widetilde{\bm A}_w^{H_l})_{\pi (\widehat{\bm Q}_{\text{com}})(ij)} - \mu \right|$ is between -2 and 2. It follows that for all large enough $d$, ${\left \{\left|(\widetilde{\bm A}_w^{G_l})_{ij} - (\widetilde{\bm A}_w^{H_l})_{\pi (\widehat{\bm Q}_{\text{com}})(ij)} - \mu \right|\ > \delta mds_2 \right \}}$ is empty and hence 
	\begin{center} 
		$1_{\left \{\left|(\widetilde{\bm A}_w^{G_l})_{ij} - (\widetilde{\bm A}_w^{H_l})_{\pi (\widehat{\bm Q}_{\text{com}})(ij)} - \mu \right|\ > \ \delta mds_2 \right \}} = 0$. 
	\end{center}
	Hence, we have as $d \rightarrow \infty$,
	\begin{align*}
		\begin{split}
			\left(\frac{1}{mds_2^2 } \right) \sum_{w=1}^d
			\sum_{l=1}^m 
			\sum_{i, j=1}^n
			E \left[ \left((\widetilde{\bm A}_w^{G_l})_{ij} - (\widetilde{\bm A}_w^{H_l})_{\pi (\widehat{\bm Q}_{\text{com}})(ij)} - \mu \right)^2 \,1_{\left \{\left|(\widetilde{\bm A}_w^{G_l})_{ij} - (\widetilde{\bm A}_w^{H_l})_{\pi (\widehat{\bm Q}_{\text{com}})(ij)} - \mu \right|\ > \ \delta mds_2 \right \}} \right] \rightarrow 0.
		\end{split}
	\end{align*} 
\end{proof}

The asymptotic distribution of $T$ under the alternative hypothesis $H_a$ has the same expression as that under the null hypothesis $H_0$. The conditions described in \eqref{cond_1} and \eqref{cond_2} impose limits on the two tails in the distribution of $(\widetilde{\bm A}_w^{G_l})_{ij} - (\widetilde{\bm A}_w^{H_l})_{\pi (\widehat{\bm Q}_{\text{com}})(ij)} - \mu$ to ensure that both tails encompass zero mass for large $d$. As a result, these conditions indicate that the variance of edge presence is constrained from below by $d^{-r}$. It is expected that the probabilities of edges within each community will display noticeable deviations from complete certainty (0 or 1), with the extent of this deviation depending on the value of $d$.

We further explore the convergence rate of $T$. Let 
\begin{center}
	$\mu = \displaystyle \frac{1}{mn(n-1)} \sum_{l=1}^m \sum_{u, v = 1}^K \sum_{i,j = 1}^n \left( \tilde{p}_{uu}^l 1_{\{i, j \in C_u\}} + \tilde{q}_{uv}^l 1_{\{i \in C_u, j \in C_v, u \neq v\}} \right)$, and \\
	$\sigma^2 = \displaystyle \frac{1}{m^2n^2(n-1)^2} \sum_{l=1}^m \sum_{u, v = 1}^K \sum_{i,j = 1}^n \left( \tilde{p}_{uu}^l (1 - \tilde{p}_{uu}^l) 1_{\{i, j \in C_u\}} + \tilde{q}_{uv}^l (1 - \tilde{q}_{uv}^l) 1_{\{i \in C_u, j \in C_v, u \neq v\}} \right)$.
\end{center}
Then the rate of approximation of $F(x) = \Pr \big(\sigma^{-1} \left(T - \mu \right) \leq x \big)$, i.e., $F(x)$ is the CDF of the normalized $T$,  can be shown by the standard normal distribution function $\Phi(y)$. Refer to Theorem 1.1 and Theorem 1.2 in \cite{deheuvels1989asymptotic}.

\begin{proposition}
	\label{prop3.4}
	(\cite{deheuvels1989asymptotic} Theorem 1.1 [p.283]) The asymptotic normality of $\sigma^{-1} \left(T - \mu \right)$, i.e., \(\sup\limits_x |F(x) - \Phi(x)| \rightarrow 0 \) as $d \rightarrow \infty$, holds if and only if as $d \rightarrow \infty$,
	\begin{center}
		$\sigma_{\dagger}^{2} = \displaystyle d
		\sum_{l=1}^m
		\sum_{u, v = 1}^K \sum_{i,j = 1}^n \left( \tilde{p}_{uu}^l (1 - \tilde{p}_{uu}^l) 1_{\{i, j \in C_u\}} + \tilde{q}_{uv}^l (1 - \tilde{q}_{uv}^l) 1_{\{i \in C_u, j \in C_v, u \neq v\}} \right) \rightarrow \infty$.
	\end{center}
\end{proposition} 
According to \Cref{prop3.4}, to attain asymptotic normality for $\sigma^{-1} \left(T - \mu \right)$, it is necessary to assume that $\sigma_{\dagger}^{2}$ diverges as $d$ approaches infinity. In simpler terms, the densities of edge presence, both within and between communities, should neither be excessively dense nor overly sparse. For the readers' convenience, concise proof is included in \Cref{appendix_a2}.

\begin{proposition}
	\label{prop3.5}
	(\cite{deheuvels1989asymptotic} Theorem 1.2 [p.283])
	For all $d \geq 1$, there exists a universal constant $c_1$ such that \[ \sup_{x} (1 + |x|^3)|F(x) - \Phi(x)| \leq c_1 \sigma_{\dagger}^{-1}. \]
\end{proposition} 
In \Cref{prop3.5}, the rate of approximation of $F(x)$ by $\Phi(x)$ is demonstrated. Specifically, the scaled deviation between the CDF of the normalized variable $T$ and the standard normal distribution is bounded by the inverse of $\sigma_{\dagger}$ at the first order.

%%%%%%%%%%%%%%%%%%%%%%%%%%%%%%%%%%%%%%%%%%%%%%%%%%%%%%%%%%%%%%%%%%%%%%%%%%

\subsection{Numerical Results} 
We represent the total number of nodes in graphs $G$ and $H$ as $n^G = \sum_{u=1}^K n_u^G$ and $n^H = \sum_{v=1}^K n_v^H$, respectively. In matrix $\bm P_{\text{com}}$, edges occur independently within each community with a probability of $p$, and between two communities with a probability of $q$. Similarly, matrix $\bm Q_{\text{com}}$ has the same block structure as $\bm P_{\text{com}}$, but the edge probabilities within each community and between communities are adjusted to $p^*$ and $q^*$ respectively, where $p^* = p + \epsilon$ and $q^* = q + \epsilon$. Under the null hypothesis, $\epsilon = 0$, while under the alternative hypothesis, we set $\epsilon > 0$.

For our analysis, we assign the values $n^G = 200$, $n^H = 100$, and choose $p$ to be 0.1 and $q$ to be 0.05. In the case of the balanced partition, we set $\tau = 0$, and for the imbalanced partition, $\tau = 0.5$. Under the null hypothesis, we set $\epsilon = 0$, while under the alternative hypothesis, we use $\epsilon \in \{0.005, 0.01\}$. One of advantages of our approach lies in its ability to discern even the most nuanced disparities in network structures. Our method exhibits outstanding performance when using a smaller $\epsilon$ value, specifically 0.005, in contrast to the previous study provided by \cite{ghoshdastidar2018practical}, which employed an $\epsilon$ value of 0.01. This suggests that our approach can effectively pinpoint minor distinctions within the intrinsic structure of two random graphs by increasing the number of bootstraps.

Within the stochastic block model, each matrix in two sequences, denoted as $(\widehat{\bm P}^{G_l}_{\text{com}})_{l \in \{1, \ldots, m\}}$ and $(\widehat{\bm Q}^{H_l}_{\text{com}})_{l \in \{1, \ldots, m\}}$, undergoes replication, resulting in a total of $d$ duplicates, where $d$ is an integer within the range of $[2, 10]$. In our analysis, we consider the sample size: $m = 2$ and different community sizes: $K \in \{2, 3, 5, 6\}$, and fix the significance level at $\alpha = 5\%$. The size of the bootstrapped adjacency matrices $n$ is chosen to be 100, 99, 100, and 96 in the specified order, corresponding to the values of $K$.

\begin{table}[H]
	\captionof{table}{$\bm P(\textbf{Type I error})$ \textbf{for balanced block sizes in bootstrapping}. Probabilities of type I errors in bootstrapping with balanced block sizes as $d$ increases from 2 to 10 for $K \in \{2,3,5,6\}$ and $\epsilon \in \{0.005, 0.01\}$.} \label{tab2}
	\centering 	
	\setlength{\tabcolsep}{0.9em}
	\renewcommand{\arraystretch}{1.1}
	
	\smallskip
	\small
	\begin{tabular}{cclccccccc}
		\cline{1-10}
		& \multicolumn{5}{c|}{$\epsilon=0.005$} 
		& \multicolumn{4}{c}{$\epsilon=0.01$}
		\\ \hline
		\multicolumn{1}{c|}{$d$}  & \multicolumn{2}{c|}{$K=2$}            & \multicolumn{1}{c|}{$K=3$} & \multicolumn{1}{c|}{$K=5$} & \multicolumn{1}{c|}{$K=6$} & \multicolumn{1}{c|}{$K=2$} & \multicolumn{1}{c|}{$K=3$} & \multicolumn{1}{c|}{$K=5$} & \multicolumn{1}{c}{$K=6$} \\ \hline
		\multicolumn{1}{c|}{2}  & \multicolumn{2}{c|}{0.052} & \multicolumn{1}{c|}{0.037}  & \multicolumn{1}{c|}{0.037}  & \multicolumn{1}{c|}{0.052}  & \multicolumn{1}{c|}{0.044}  & \multicolumn{1}{c|}{0.050}  & \multicolumn{1}{c|}{0.055}  &  0.053 \\ 
		\multicolumn{1}{c|}{3}  & \multicolumn{2}{c|}{0.045} & \multicolumn{1}{c|}{0.051}  & \multicolumn{1}{c|}{0.049}  & \multicolumn{1}{c|}{0.056}  & \multicolumn{1}{c|}{0.058}  & \multicolumn{1}{c|}{0.045}  & \multicolumn{1}{c|}{0.053}  &  0.056 \\ 
		\multicolumn{1}{c|}{4}  & \multicolumn{2}{c|}{0.053} & \multicolumn{1}{c|}{0.055}  & \multicolumn{1}{c|}{0.051}  & \multicolumn{1}{c|}{0.053}  & \multicolumn{1}{c|}{0.047}  & \multicolumn{1}{c|}{0.057}  & \multicolumn{1}{c|}{0.049}  &  0.051 \\ 
		\multicolumn{1}{c|}{5}  & \multicolumn{2}{c|}{0.042} & \multicolumn{1}{c|}{0.057}  & \multicolumn{1}{c|}{0.055}  & \multicolumn{1}{c|}{0.047}  & \multicolumn{1}{c|}{0.053}  & \multicolumn{1}{c|}{0.067}  & \multicolumn{1}{c|}{0.057}  &  0.066 \\ 
		\multicolumn{1}{c|}{6}  & \multicolumn{2}{c|}{0.056} & \multicolumn{1}{c|}{0.048}  & \multicolumn{1}{c|}{0.044}  & \multicolumn{1}{c|}{0.048}  & \multicolumn{1}{c|}{0.043}  & \multicolumn{1}{c|}{0.043}  & \multicolumn{1}{c|}{0.042}  &  0.052 \\ 
		\multicolumn{1}{c|}{7}  & \multicolumn{2}{c|}{0.056} & \multicolumn{1}{c|}{0.054}  & \multicolumn{1}{c|}{0.052}  & \multicolumn{1}{c|}{0.053}  & \multicolumn{1}{c|}{0.049}  & \multicolumn{1}{c|}{0.049}  & \multicolumn{1}{c|}{0.058}  &  0.054 \\ 
		\multicolumn{1}{c|}{8}  & \multicolumn{2}{c|}{0.041} & \multicolumn{1}{c|}{0.065}  & \multicolumn{1}{c|}{0.058}  & \multicolumn{1}{c|}{0.058}  & \multicolumn{1}{c|}{0.049}  & \multicolumn{1}{c|}{0.050}  & \multicolumn{1}{c|}{0.054}  &  0.058 \\ 
		\multicolumn{1}{c|}{9}  & \multicolumn{2}{c|}{0.054} & \multicolumn{1}{c|}{0.059}  & \multicolumn{1}{c|}{0.049}  & \multicolumn{1}{c|}{0.066}  & \multicolumn{1}{c|}{0.054}  & \multicolumn{1}{c|}{0.044}  & \multicolumn{1}{c|}{0.053}  &  0.042 \\ 
		\multicolumn{1}{c|}{10} & \multicolumn{2}{c|}{0.051} & \multicolumn{1}{c|}{0.050}  & \multicolumn{1}{c|}{0.054}  & \multicolumn{1}{c|}{0.059}  & \multicolumn{1}{c|}{0.053}  & \multicolumn{1}{c|}{0.051}  & \multicolumn{1}{c|}{0.045}  &  0.053 \\ \hline
	\end{tabular}    
\end{table}

\begin{table}[H]
	\captionof{table}{\textbf{Test power for balanced block sizes in bootstrapping}. Test power is illustrated as $d$ increases from 2 to 10, with balanced block sizes and considering $K \in \{2,3,5,6\}$ and $\epsilon \in \{0.005, 0.01\}$.} \label{tab3}
	\centering 	
	\setlength{\tabcolsep}{0.9em}
	\renewcommand{\arraystretch}{1.1}
	
	\smallskip
	\small
	\begin{tabular}{cclccccccc}
		\cline{1-10}
		& \multicolumn{5}{c|}{$\epsilon=0.005$} 
		& \multicolumn{4}{c}{$\epsilon=0.01$}  \\ \hline
		\multicolumn{1}{c|}{$d$}  & \multicolumn{2}{c|}{$K=2$} & \multicolumn{1}{c|}{$K=3$} & \multicolumn{1}{c|}{$K=5$} & \multicolumn{1}{c|}{$K=6$} & \multicolumn{1}{c|}{$K=2$} & \multicolumn{1}{c|}{$K=3$} & \multicolumn{1}{c|}{$K=5$} & \multicolumn{1}{c}{$K=6$} \\ \hline
		\multicolumn{1}{c|}{2}  & \multicolumn{2}{c|}{0.652} & \multicolumn{1}{c|}{0.676}  & \multicolumn{1}{c|}{0.702}  & \multicolumn{1}{c|}{0.712}  & \multicolumn{1}{c|}{0.954}  & \multicolumn{1}{c|}{0.975}  & \multicolumn{1}{c|}{0.973}  &  0.978 \\ 
		\multicolumn{1}{c|}{3}  & \multicolumn{2}{c|}{0.699} & \multicolumn{1}{c|}{0.720}  & \multicolumn{1}{c|}{0.768}  & \multicolumn{1}{c|}{0.761}  & \multicolumn{1}{c|}{0.981}  & \multicolumn{1}{c|}{0.990}  & \multicolumn{1}{c|}{0.982}  &  0.991 \\ 
		\multicolumn{1}{c|}{4}  & \multicolumn{2}{c|}{0.765} & \multicolumn{1}{c|}{0.772}  & \multicolumn{1}{c|}{0.807}  & \multicolumn{1}{c|}{0.814}  & \multicolumn{1}{c|}{0.989}  & \multicolumn{1}{c|}{0.993}  & \multicolumn{1}{c|}{0.993}  &  0.996 \\ 
		\multicolumn{1}{c|}{5}  & \multicolumn{2}{c|}{0.793} & \multicolumn{1}{c|}{0.832}  & \multicolumn{1}{c|}{0.818}  & \multicolumn{1}{c|}{0.839}  & \multicolumn{1}{c|}{0.991}  & \multicolumn{1}{c|}{0.992}  & \multicolumn{1}{c|}{0.997}  &  0.993 \\ 
		\multicolumn{1}{c|}{6}  & \multicolumn{2}{c|}{0.807} & \multicolumn{1}{c|}{0.832}  & \multicolumn{1}{c|}{0.879}  & \multicolumn{1}{c|}{0.866}  & \multicolumn{1}{c|}{0.991}  & \multicolumn{1}{c|}{0.997}  & \multicolumn{1}{c|}{0.999}  &  0.999 \\ 
		\multicolumn{1}{c|}{7}  & \multicolumn{2}{c|}{0.826} & \multicolumn{1}{c|}{0.860}  & \multicolumn{1}{c|}{0.854}  & \multicolumn{1}{c|}{0.896}  & \multicolumn{1}{c|}{0.992}  & \multicolumn{1}{c|}{0.996}  & \multicolumn{1}{c|}{0.996}  &  0.999 \\ 
		\multicolumn{1}{c|}{8}  & \multicolumn{2}{c|}{0.847} & \multicolumn{1}{c|}{0.869}  & \multicolumn{1}{c|}{0.881}  & \multicolumn{1}{c|}{0.895}  & \multicolumn{1}{c|}{0.994}  & \multicolumn{1}{c|}{0.997}  & \multicolumn{1}{c|}{0.993}  &  0.998 \\ 
		\multicolumn{1}{c|}{9}  & \multicolumn{2}{c|}{0.851} & \multicolumn{1}{c|}{0.888}  & \multicolumn{1}{c|}{0.887}  & \multicolumn{1}{c|}{0.914}  & \multicolumn{1}{c|}{0.998}  & \multicolumn{1}{c|}{0.998}  & \multicolumn{1}{c|}{0.996}  &  1.000 \\ 
		\multicolumn{1}{c|}{10} & \multicolumn{2}{c|}{0.836} & \multicolumn{1}{c|}{0.902}  & \multicolumn{1}{c|}{0.909}  & \multicolumn{1}{c|}{0.904}  & \multicolumn{1}{c|}{0.997}  & \multicolumn{1}{c|}{0.997}  & \multicolumn{1}{c|}{0.999}  &  0.999 \\ 
		\hline
	\end{tabular}   
\end{table}

When $\epsilon$ is set to 0.005, \Cref{tab2} demonstrates the probability of type I error for balanced block sizes computed from 1000 independent experiment runs. Under the null model, the null hypothesis is rejected at a rate of approximately $\alpha = 5\%$. Conversely, under the alternative model, a high test power (close to 1) is desirable, as illustrated in \Cref{tab3}. Notably, when $\epsilon$ is increased to 0.01, our test statistic exhibits superior test power compared to the test statistic, \textbf{Asymp-Normal}, proposed in \cite{ghoshdastidar2018practical}. Additional simulation results for scenarios with imbalanced block sizes can be found in \Cref{appendix_a3}.

\begin{center}
	\includegraphics[scale=.69]{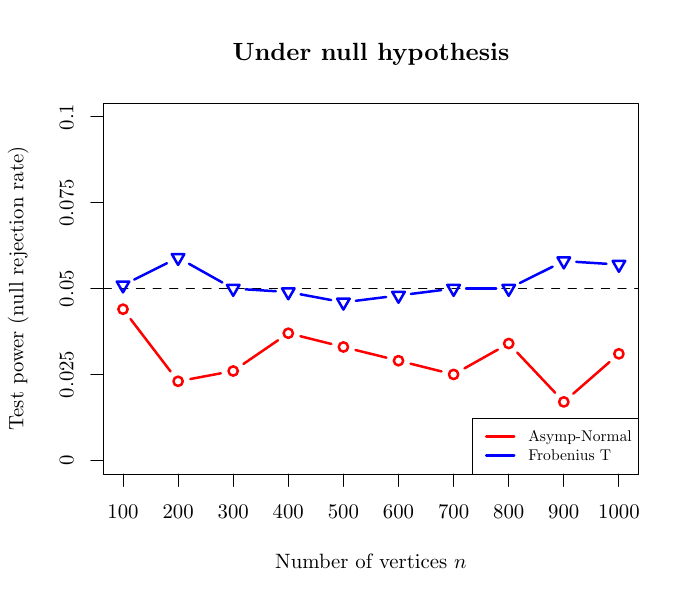} 
	\includegraphics[scale=.69]{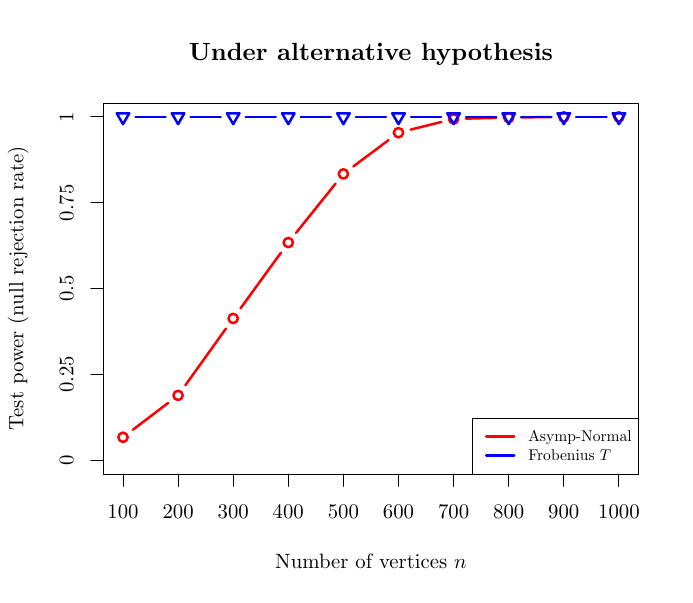} 
	\captionof{figure}{Power of \textbf{Asymp-Normal} and \textbf{Frobenius} $T$ for increasing number of vertices $n$ (and $n^G$), and for $m = 2$, $K = 2$, $\epsilon = 0.04$, and $\tau = 0$. The dotted line for case of null hypothesis corresponds to the significance level of 5\%.} \label{fig1}
\end{center}

To evaluate the reliability and effectiveness of our test statistic $T$ compared to \textbf{Asymp-Normal}, we analyze their probabilities of type I error and test power. The assessment is carried out under the conditions of $m = 2$, $K = 2$, $\epsilon = 0.04$, and $\tau = 0$. For \textbf{Asymp-Normal}, we systematically vary the network size $n$ from 100 to 1000 in increments of 100, maintaining consistency across both graphs $G$ and $H$. To minimize potential variations arising from parameter choices, we independently vary $n^G$, ranging from 100 to 1000 in increments of 100, and set $n^H$ to always be 100 units larger than $n^G$, all while maintaining a constant value of $d$ at 2.

In \Cref{fig1}, under the null hypothesis, the probability of type I error for our Frobenius test statistic $T$ fluctuates around 5\%, whereas \textbf{Asymp-Normal} exhibits a comparatively smaller probability of type I error. However, for small network sizes, Frobenius $T$ demonstrates 100\% test power, while \textbf{Asymp-Normal} shows significantly lower test power.

%%%%%%%%%%%%%%%%%%%%%%%%%%%%%%%%%%%%%%%%%%%%%%%%%%%%%%%%%%%%%%%%%%

\section{Application to fMRI data} \label{app_sec}

In neuroscience and cognitive research, the analysis of functional Magnetic Resonance Imaging (fMRI) data plays a pivotal role in understanding brain network effects and function \citep{bullmore2009complex, smith2009correspondence, van2010exploring, power2011functional}. We examine fMRI data collected from the StarPlus study using our proposed Frobenius test statistic to determine the impact of study conditions on brain function and connectivity.

The experiment is structured as a series of trials, with each measurement taken in regular intervals over the course of these trials for each participant. During specific intervals, participants are either given the opportunity to rest or are instructed to fixate their gaze on a designated point displayed on the screen. In other trials, subjects are presented with both a picture and a sentence, where they are tasked with determining whether the sentence accurately describes the picture. These trials are carefully designed to include the sequential presentation of the sentence and picture, with half of the trials displaying the picture first, and the other half introducing the sentence initially.

We adopt an approach by analyzing the during the first stimulus (images 1 - 8), the second stimulus (images 17 - 24), and the control group (images 33 - 54) for each subject, with the exception of subject 05675, whose control group comprises images 33 - 53. This selective data extraction allows for a meticulous exploration of brain responses during different stimuli, offering insights into whether our cognitive apparatus operates differently when subjected to various visual and linguistic inputs.

For every image, a column-by-column matching process is implemented based on Regions of Interest (ROI's). Subsequently, data for each subject is averaged over multiple voxels within the same ROI for three networks: the first stimulus, the second stimulus, and the control group. This averaging procedure results in the creation of three matrices per participant corresponding to each period of the trial.

The correlation matrix is derived from the pairwise comparison between rows of ROI brain functional activity matrix for each of the three trial periods for each participant. These correlation matrices across the ROIs are transformed to adjacency matrices using a rank-determined threshold on their magnitude, which ensures the identical edge density of three groups for each subject. Subsequently, random graphs are constructed using these adjacency matrices.

In the subsequent phase, community-wise edge probability matrices are derived using a sequential community detection algorithm. This process commences with hypothesis testing for $H_0: K=1$ vs. $H_a: K>1$, applied to all three networks. If all three networks reject $H_0$, the procedure is repeated for $H_0: K=2$ vs. $H_a: K>2$. This iterative process continues until one of the networks fails to reject $H_0$. The value of $K$ resulting from this outcome is then applied to all three networks. The matrices derived are denoted as $\bm P^1_{\text{com}}$, $\bm P^2_{\text{com}}$, and $\bm P^c_{\text{com}}$, representing community-wise edge probability matrices for the first stimulus, the second stimulus, and the control group, respectively.

This paves the way for three following sets of hypothesis testing, each tailored to evaluate the specific properties and characteristics of the networks generated through this process.

\begin{center}
	$H_0: \bm P^1_{\text{com}} = \bm P^2_{\text{com}}$ vs $H_a: \bm P^1_{\text{com}} \neq \bm P^2_{\text{com}}$, \\
	$H_0: \bm P^1_{\text{com}} = \bm P^c_{\text{com}}$ vs $H_a: \bm P^1_{\text{com}} \neq \bm P^c_{\text{com}}$, \\
	$H_0: \bm P^2_{\text{com}} = \bm P^c_{\text{com}}$ vs $H_a: \bm P^2_{\text{com}} \neq \bm P^c_{\text{com}}$. \\
\end{center}

For each participant, we employ an appropriate threshold to create adjacency matrices for all three networks. We use the sequential community detection algorithm to determine the number of communities. The network size for bootstrapping is fixed at 24, matching the number of ROI's. Subsequently, taking into account the relatively small size of bootstrapped samples, we introduce normal random errors with a standard deviation of $\tau = 0$ to establish community sizes in bootstrapping.

\Crefrange{fig2}{fig4} show the graph of Subject 05710 with labeled ROI's for the first stimulus, the second stimulus, and the control group.
\begin{center}
	\includegraphics[scale=.5]{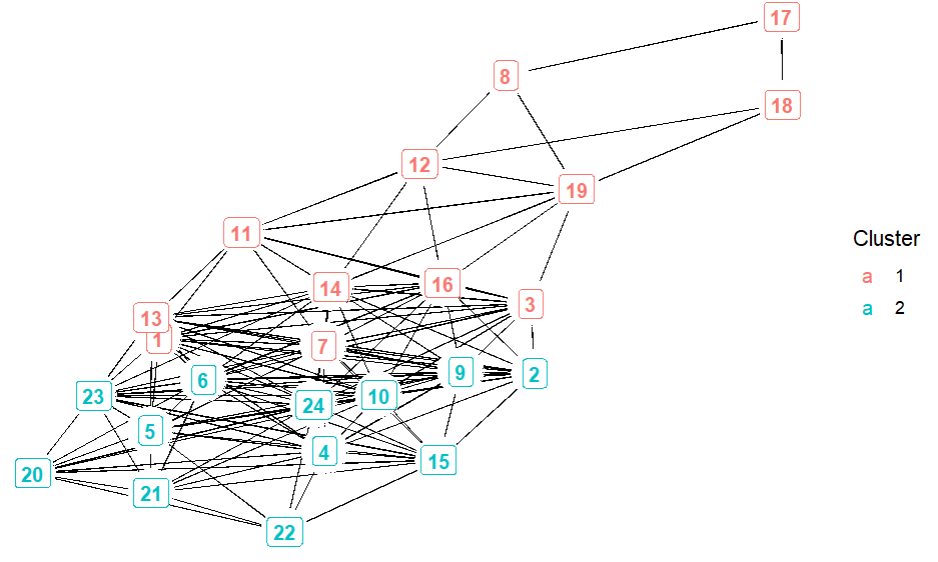} 
	\captionof{figure}{Subject 05710 ROI's labeled graph for Stimulus 1} \label{fig2}
\end{center}

\begin{center}
	\includegraphics[scale=.5]{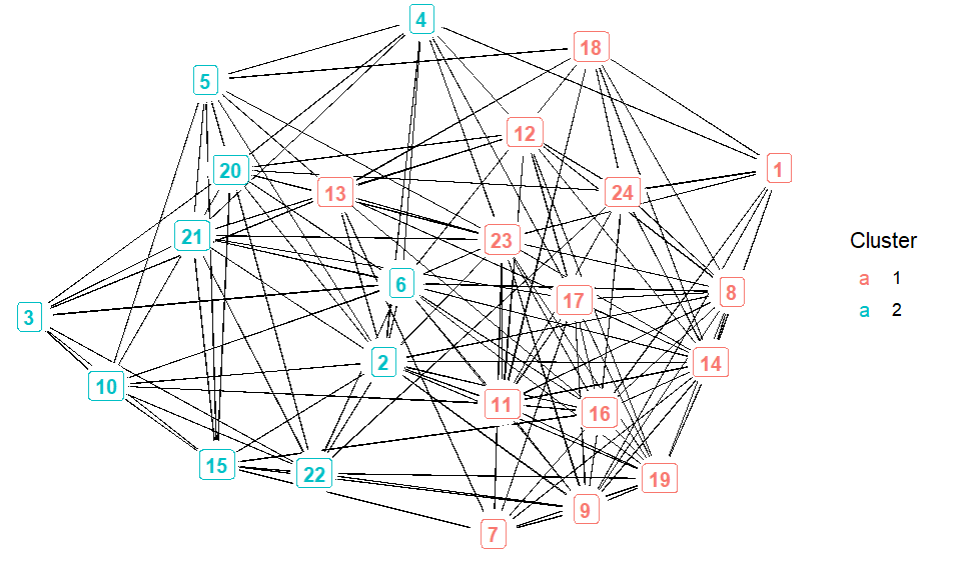} 
	\captionof{figure}{Subject 05710 ROI's labeled graph for Stimulus 2} \label{fig3}
\end{center}

\begin{center}
	\includegraphics[scale=.5]{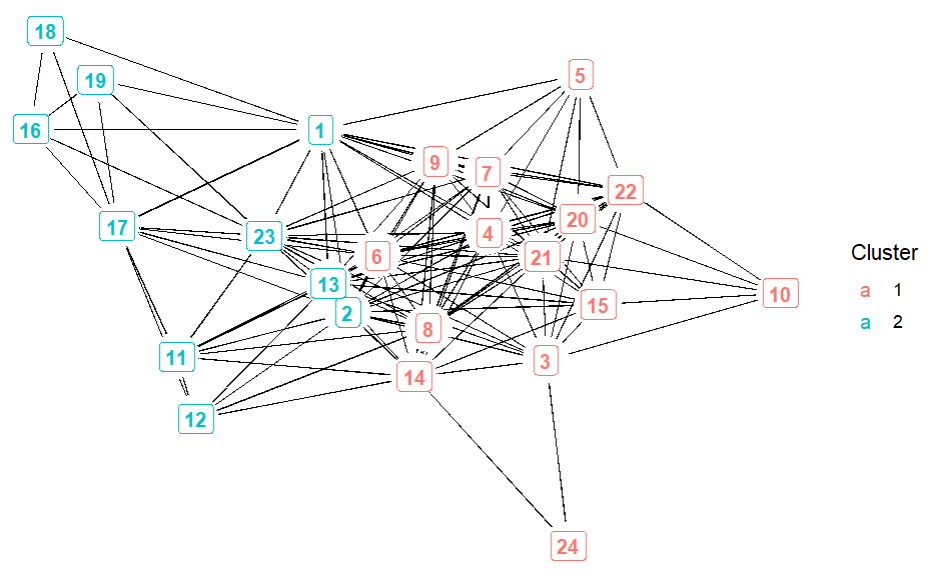} 
	\captionof{figure}{Subject 05710 ROI's labeled graph for Control group} \label{fig4}
\end{center}

\begin{center}
	\captionof{table}{\textbf{Results of hypothesis testing for six subjects across three sets of comparisons with the value of $d = 100$.} Remark: $p$ values are shown in the parenthesis, with the notation that $*$ represents $p \leq 0.05$, $**$ represents $p \leq 0.01$, $***$ represents $p \leq 0.001$, and $****$ represents $p \leq 0.0001$.} \label{tab4}
	\centering 	
	\setlength{\tabcolsep}{0.9em}
	\renewcommand{\arraystretch}{1.1}
	
	\smallskip
	\small
	\begin{tabular}{ l || l | l | l | l }   
		\hline
		\text{Subject ID} & $\hat{K}$ & \text{Stimulus 1 vs Stimulus 2} & \text{Stimulus 1 vs Control} & \text{Stimulus 2 vs Control} \\
		\hline
		04799 & 2 & \text{Reject $H_0 (****)$} & \text{Fail to Reject $H_0$} & \text{Fail to Reject $H_0$} \\ 
		04820 & 2 & \text{Reject $H_0 (***)$} & \text{Fail to Reject $H_0$} & \text{Reject $H_0 (****)$} \\ 
		04847 & 2 & \text{Reject $H_0 (****)$} & \text{Reject $H_0 (****)$} & \text{Fail to Reject $H_0$} \\ 
		05675 & 2 & \text{Reject $H_0 (****)$} & \text{Reject $H_0 (****)$} & \text{Fail to Reject $H_0$} \\
		05680 & 2 & \text{Reject $H_0 (*)$} & \text{Reject $H_0 (****)$} & \text{Fail to Reject $H_0$} \\ 
		05710 & 2 & \text{Reject $H_0 (****)$} & \text{Fail to Reject $H_0$} & \text{Reject $H_0 (**)$} \\ 
		\hline
	\end{tabular}
\end{center}

In light of the findings displayed in \Cref{tab4}, for most subjects there are significant differences in brain activity patterns when subjects are exposed to sentences and pictures. Nevertheless, the differentiation of ROI's during a stimulus versus a rest period varies among individuals.

This fMRI dataset reveals an interesting characteristic. If you intend to train a classifier to distinguish between when the subject is viewing a picture or a sentence, you are likely to achieve the highest accuracy by restricting your analysis to the following ROI's: `CALC,' `LIPL,' `LT,' `LTRIA,' `LOPER,' `LIPS,' and `LDLPFC'. The ratios of the mentioned ROI's among communities in each graph for each subject are presented in \Cref{tab5}. This suggests higher importance in the activitation of these ROIs for sentence and picture comprehension.

\begin{center}
	\captionof{table}{Ratios of significant ROI's among communities in each random graph for each subject} \label{tab5}
	\centering 	
	\setlength{\tabcolsep}{0.9em}
	\renewcommand{\arraystretch}{1.1}
	
	\smallskip
	\small
	\begin{tabular}{ l || l | l | l }   
		\hline
		\text{Subject ID} & \text{Stimulus 1} & \text{Stimulus 2} & \text{Control} \\
		\hline
		04799 & 4:3 & 4:3 & 4:3 \\ 
		04820 & 4:3 & 4:3 & 4:3 \\ 
		04847 & 6:1 & 4:3 & 5:2 \\ 
		05675 & 4:3 & 4:3 & 4:3 \\
		05680 & 6:1 & 4:3 & 5:2 \\ 
		05710 & 4:3 & 4:3 & 4:3 \\ 
		\hline
	\end{tabular}
\end{center}

Within each graph, every community comprises some of these important ROI's, indicating the rationality of the community structures in all the graphs obtained by the sequential community detection algorithm.

%%%%%%%%%%%%%%%%%%%%%%%%%%%%%%%%%%%%%%%%%%%%%%%%%%%%%%%%%%%%%%%%%%%%%%%%%%%%%%%%%%%%%%%%

\section{Discussion} \label{dis_sec}

This article presents a novel approach for comparing two-sample large graphs with distinct vertex sets. The approach includes generating bootstrapped adjacency matrices based on estimated community-wise edge probability matrices, extracting density information regarding the presence of edges within and between communities from the samples. A Frobenius test statistic is constructed, demonstrating convergence to a normal distribution with an increasing number of bootstraps. Its superior capability to identify subtle differences distinguishes it from existing methods. Moreover, its broader potential applications stem from its intrinsic design for random graphs of unequal sizes.

The methods for testing two-sample large random graphs have seen rapid development, and their applications in many fields are intriguing and exciting, like neuroscience and social science. However, there remains ample room for researchers to explore innovative methods and broader applications. For instance, although the proposed test statistic is applicable to directed graphs, it does not encompass both in-degrees and out-degrees, which are fundamental measures in understanding the flow and connectivity within such networks. This inspires future research to consider incorporating these structural features into the study of directed graphs. We could also anticipate the integration of network data with advanced machine learning techniques, such as network data clustering.

 %%%%%%%%%%%%%%%%%%%%%%%%%%%%%%%%%%%%%%%%%%%%%%%%%%%%%%%%%%%%%%%%%%%%%%%%%%%%%%%%%%%%%%%

\renewcommand{\thesection}{\Alph{section}.\arabic{section}}
\appendix
\section{Appendix} \label{appendix}
\subsection{Proof of \Cref{prop3.1}} \label{appendix_a1}
\begin{proof}
	Under the null hypothesis, the optimal permutation $\pi$ can be found so that $(\widetilde{\bm A}_w^{G_l})_{ij}$ and  $(\widetilde{\bm A}_w^{H_l})_{\pi(\widehat{\bm Q}_{\text{com}})(ij)}$ have the same distribution. For notation simplicity, we drop the notation for the optimal permutation $\pi$ and we simply denote $(\widetilde{\bm A}_w^{H_l})_{\pi(\widehat{\bm Q}_{\text{com}})(ij)}$ as  $(\widetilde{\bm A}_w^{H_l})_{ij}$ in the rest of the proof. Using the general form of the simple inequality $\sqrt{a+ b+c+d} \leq \sqrt{a+b} +\sqrt{c+d}$, for nonnegative values $a, b, c, d,$  we estimate the square root of the term in the conclusion of the proposition as 
	\begin{align} \label{K28April23.4}
		\begin{split} 
			& \left( \frac{1}{mds_1^2} \right)
			\sum_{w=1}^d
			\sum_{l=1}^m 
			\sum_{i, j=1}^n
			E\left[ \left((\widetilde{\bm A}_w^{G_l})_{ij} - (\widetilde{\bm A}_w^{H_l})_{ij}\right)^2 \, 1_{ \left\{\left|(\widetilde{\bm A}_w^{G_l})_{ij} - (\widetilde{\bm A}_w^{H_l})_{ij}\right| \ > \ \delta md s_1 \right\} } \right] \\ 
			& \hspace{0.5cm} = \sum_{i,j=1}^n \left( \frac{1}{mds_1^2}\right)\,
			\sum_{w=1}^d\,
			\sum_{l=1}^m\,  E\left[\left((\widetilde{\bm A}_w^{G_l})_{ij} - (\widetilde{\bm A}_w^{H_l})_{ij}\right)^2 \, 1_{\left\{\left|(\widetilde{\bm A}_w^{G_l})_{ij} - (\widetilde{\bm A}_w^{H_l})_{ij}\right| \ >\  \delta mds_1\right\}} \right]. 
		\end{split}
	\end{align}
	We now focus on the summands in the above. First, for the term $\big((\widetilde{\bm A}_w^{G_l})_{ij} - (\widetilde{\bm A}_w^{H_l})_{ij}\big)^2$ in the summand, we use  the inequality $(a+b)^2 = a^2+2ab+b^2 \leq 2(a^2+b^2),$ for all real values $a, b$, to  see that 
	\begin{align} \label{K28April23.5}
		\begin{split} 
			& \left((\widetilde{\bm A}_w^{G_l})_{ij} - (\widetilde{\bm A}_w^{H_l})_{ij}\right)^2 = \left(\left((\widetilde{\bm A}_w^{G_l})_{ij} -\mu_{ij} \right)+\left(\mu_{ij}  - (\widetilde{\bm A}_w^{H_l})_{ij}\right) \right)^2 \\
			& \hspace{0.5cm} \leq 2 \left(\left((\widetilde{\bm A}_w^{G_l})_{ij} -\mu_{ij} \right)^2 +\left((\widetilde{\bm A}_w^{H_l})_{ij}- \mu_{ij}\right)^2 \right).
		\end{split}
	\end{align}
	Second, we focus on the indicator function in \eqref{K28April23.4}. By the Triangle Inequality, 
	\begin{align*} 
		\begin{split} 
			\left| (\widetilde{\bm A}_w^{G_l})_{ij} - (\widetilde{\bm A}_w^{H_l})_{ij} \right|  \ \leq  \left| (\widetilde{\bm A}_w^{G_l})_{ij} -\mu_{ij} \right| + \left| (\widetilde{\bm A}_w^{H_l})_{ij} -\mu_{ij} \right|.
		\end{split}
	\end{align*}
	Hence, if $\left| (\widetilde{\bm A}_w^{G_l})_{ij} - (\widetilde{\bm A}_w^{H_l})_{ij} \right| > \delta mds_1$, then either 
	\begin{align*} 
		\begin{split} 
			\left|(\widetilde{\bm A}_w^{G_l})_{ij} - \mu_{ij} \right| \  > \  \frac{\delta}{2}  mds_1 \   \ \  \mbox{or } \ \ \ \left|(\widetilde{\bm A}_w^{H_l})_{ij} - \mu_{ij} \right| \  > \ \frac{\delta}{2} mds_1.
		\end{split}
	\end{align*}
	In other words, the set $\left\{\left|(\widetilde{\bm A}_w^{G_l})_{ij} - (\widetilde{\bm A}_w^{H_l})_{ij}\right| > \delta mds_1\right\}$ of the indicator function in \eqref{K28April23.4} is contained in the union 
	\begin{align*} 
		\begin{split} 
			\left\{\left|(\widetilde{\bm A}_w^{G_l})_{ij} - \mu_{ij} \right| > \frac{\delta}{2} mds_1\right\}  \cup  
			\left\{\left|(\widetilde{\bm A}_w^{H_l})_{ij} - \mu_{ij} \right| >  \frac{\delta}{2} mds_1\right\}. 
		\end{split}
	\end{align*}
	Thus we have 
	\begin{align} \label{K28April23.6}
		\begin{split}  
			1_{\left\{\left|(\widetilde{\bm A}_w^{G_l})_{ij} - (\widetilde{\bm A}_w^{H_l})_{ij}\right| \ >\  \delta mds_1\right\}}
			\leq 1_{\left\{\left|(\widetilde{\bm A}_w^{G_l})_{ij} - \mu_{ij} \right| > \frac{\delta}{2} mds_1\right\} } \ + \ 1_{ \left\{\left|(\widetilde{\bm A}_w^{H_l})_{ij} - \mu_{ij} \right| >  \frac{\delta}{2} mds_1 \right\}}.
		\end{split}
	\end{align}
	
	Lastly, we use \eqref{K28April23.5} and \eqref{K28April23.6} to estimate the summand in \eqref{K28April23.4} as follows,  
	\begin{align} 
		\begin{split}
			& \left((\widetilde{\bm A}_w^{G_l})_{ij} - (\widetilde{\bm A}_w^{H_l})_{ij}\right)^2 \, 1_{\left\{\left|(\widetilde{\bm A}_w^{G_l})_{ij}- (\widetilde{\bm A}_w^{H_l})_{ij}\right| \ >\  \delta mds_1 \right\}} \\ 
			& \hspace{0.2cm} \leq 2 \left(\left((\widetilde{\bm A}_w^{G_l})_{ij} - \mu_{ij} \right)^2 + \left((\widetilde{\bm A}_w^{H_l})_{ij}- \mu_{ij}\right)^2 \right) 
			\left(1_{\left\{\left|(\widetilde{\bm A}_w^{G_l})_{ij} - \mu_{ij} \right| > \frac{\delta}{2} mds_1 \right\} } \ + \ 1_{\left\{\left|(\widetilde{\bm A}_w^{H_l})_{ij} - \mu_{ij} \right| >  \frac{\delta}{2}  mds_1 \right\}} \right)  \\
			&  \hspace{0.2cm} = 2 \left((\widetilde{\bm A}_w^{G_l})_{ij} -\mu_{ij} \right)^2\,1_{\left\{\left|(\widetilde{\bm A}_w^{G_l})_{ij}- \mu_{ij} \right| > \frac{\delta}{2} mds_1 \right\}}\ + \ 2 \left((\widetilde{\bm A}_w^{G_l})_{ij} -\mu_{ij} \right)^2\, 1_{\left\{\left|(\widetilde{\bm A}_w^{H_l})_{ij} - \mu_{ij} \right| >  \frac{\delta}{2}  mds_1 \right\}} \\ \notag
			& \hspace{0.5cm} \ + \ 2 \left((\widetilde{\bm A}_w^{H_l})_{ij}- \mu_{ij} \right)^2\,1_{\left\{\left|(\widetilde{\bm A}_w^{G_l})_{ij} - \mu_{ij} \right| > \frac{\delta}{2} mds_1 \right\}}\ + \ 2 \left((\widetilde{\bm A}_w^{H_l})_{ij}- \mu_{ij} \right)^2\,1_{\left\{\left|(\widetilde{\bm A}_w^{H_l})_{ij} - \mu_{ij} \right| > \frac{\delta}{2} mds_1 \right\} }.
		\end{split}
	\end{align}
	Hence we get
	\begin{align} \label{K28April23.7}
		\begin{split} 
			E & \left[\left((\widetilde{\bm A}_w^{G_l})_{ij} - (\widetilde{\bm A}_w^{H_l})_{ij}\right)^2 \, 1_{\left\{\left|(\widetilde{\bm A}_w^{G_l})_{ij} - (\widetilde{\bm A}_w^{H_l})_{ij}\right| \ >\  \delta mds_1 \right\}} \right] \\
			& \leq 
			2E \left[\left((\widetilde{\bm A}_w^{G_l})_{ij} -\mu_{ij} \right)^2\, 1_{\left\{\left|(\widetilde{\bm A}_w^{G_l})_{ij} - \mu_{ij} \right| > \frac{\delta}{2} mds_1 \right\} }\right]  \\
			& \hspace{0.5cm} + \ 2E \left[\left((\widetilde{\bm A}_w^{G_l})_{ij} - \mu_{ij} \right)^2\, 1_{\left\{\left|(\widetilde{\bm A}_w^{H_l})_{ij} - \mu_{ij} \right|  >  \frac{\delta}{2} mds_1 \right\} }\right] \\
			& \hspace{0.5cm} + \ 2E \left[\left((\widetilde{\bm A}_w^{H_l})_{ij}- \mu_{ij}\right)^2\,1_{\left\{\left|(\widetilde{\bm A}_w^{G_l})_{ij} - \mu_{ij} \right| > \frac{\delta}{2} mds_1 \right\} }\right] \\
			& \hspace{0.5cm} + \ 2E \left[\left((\widetilde{\bm A}_w^{H_l})_{ij}- \mu_{ij}\right)^2\,1_{ \left\{\left|(\widetilde{\bm A}_w^{H_l})_{ij} - \mu_{ij} \right| >  \frac{\delta}{2} mds_1 \right\} }\right].
		\end{split}
	\end{align}
	Under the null hypothesis, 
	\begin{align*} 
		\begin{split}
			\left\{\left|(\widetilde{\bm A}_w^{G_l})_{ij} - \mu_{ij}\right| \ >\  \frac{\delta}{2} mds_1 \right\} \mbox{ \ and \ }
			\left\{\left|(\widetilde{\bm A}_w^{H_l})_{ij} - \mu_{ij} \right| > \frac{\delta}{2} mds_1 \right\} 
		\end{split}
	\end{align*}
	have the same distribution, and hence we can rewrite \eqref{K28April23.7} as 
	\begin{align*} 
		\begin{split}
			E & \left[\left((\widetilde{\bm A}_w^{G_l})_{ij} - (\widetilde{\bm A}_w^{H_l})_{ij}\right)^2 \, 1_{\left\{\left|(\widetilde{\bm A}_w^{G_l})_{ij} - (\widetilde{\bm A}_w^{H_l})_{ij}\right| \ >\  \delta mds_1 \right\}} \right] \\
			& \leq 4E \left[\left((\widetilde{\bm A}_w^{G_l})_{ij} -\mu_{ij} \right)^2\, 1_{\left\{\left|(\widetilde{\bm A}_w^{G_l})_{ij} - \mu_{ij} \right| > \frac{\delta}{2} mds_1 \right \} }\right] + 4E \bigg[\big((\widetilde{\bm A}_w^{H_l})_{ij}- \mu_{ij}\big)^2\,1_{ \big\{\big|(\widetilde{\bm A}_w^{H_l})_{ij} - \mu_{ij} \big| >  \frac{\delta}{2} mds_1\big\}} \bigg].
		\end{split}
	\end{align*}
	Using this inequality, we continue the estimation  the last term in \eqref{K28April23.4} as follows.
	\begin{align*} 
		\begin{split}
			& \left( \frac{1}{ mds_1^2}\right)\, \sum_{l=1}^m\, \sum_{w=1}^d\, E\left[\left((\widetilde{\bm A}_w^{G_l})_{ij}- (\widetilde{\bm A}_w^{H_l})_{ij}\right)^2\, 1_{\left\{\left|(\widetilde{\bm A}_w^{G_l})_{ij} - (\widetilde{\bm A}_w^{H_l})_{ij}\right| \  > \ \delta mds_1 \right\}} \right]  \\ 
			& \hspace{0.5cm} \leq \left(\frac{4}{ mds_1^2} \right) \, \sum_{l=1}^m\, \sum_{w=1}^d\, E\left[\left((\widetilde{\bm A}_w^{G_l})_{ij} -\mu_{ij} \right)^2\,1_{\left\{\left|(\widetilde{\bm A}_w^{G_l})_{ij} - \mu_{ij} \right| > \frac{\delta}{2} mds_1 \right\}} \right] \\ 
			& \hspace{1cm} \ + \ 
			\left(\frac{4}{ mds_1^2} \right)\, \sum_{l=1}^m\, \sum_{w=1}^d\, E\left[\left((\widetilde{\bm A}_w^{H_l})_{ij}- \mu_{ij}\right)^2\,1_{ \left\{\left|(\widetilde{\bm A}_w^{H_l})_{ij} - \mu_{ij} \right| > \frac{\delta}{2} mds_1 \right\}} \right].
		\end{split}
	\end{align*}
	Under the null hypothesis, if $i, j$ are in $C_u$, then $\mu_{ij} = \hat{p}_{uu}$ and $\sigma_{ij}^2 = \hat{p}_{uu}(1-\hat{p}_{uu})$. Hence, $s_1^2 = 2n(n-1)\hat{p}_{uu}(1-\hat{p}_{uu})$, and the above inequality can be rewritten as 
	\begin{align} \label{K28April23.8}
		\begin{split}
			& \left(\frac{4}{ mds_1^2} \right) \, \sum_{l=1}^m\, \sum_{w=1}^d\,  E\left[\left((\widetilde{\bm A}_w^{G_l})_{ij} - \hat{p}_{uu} \right)^2\, 1_{\left\{\left|(\widetilde{\bm A}_w^{G_l})_{ij} - \hat{p}_{uu} \right| > \frac{\delta}{2} mds_1 \right\} } \right] \\ 
			& \hspace{0.5cm} \ + \ 
			\left( \frac{4}{ mds_1^2}\right)\, \sum_{l=1}^m\, \sum_{w=1}^d\, E\left[\left((\widetilde{\bm A}_w^{H_l})_{ij} - \hat{p}_{uu} \right)^2\, 1_{\left\{\left|(\widetilde{\bm A}_w^{H_l})_{ij} - \hat{p}_{uu} \right| > \frac{\delta}{2} mds_1 \right\}}\right]. 
		\end{split}
	\end{align}
	
	Similarly, if $i$ is in $C_u$ and $j$ is in $C_v$ with $u \neq v$, then $\mu_{ij} = \hat{q}_{uv}$ and $\sigma_{ij}^2 = \hat{q}_{uv}(1-\hat{q}_{uv})$. To finish the whole proof, we note that $(\widetilde{\bm A}_w^{G_l})_{ij}$ and $(\widetilde{\bm A}_w^{H_l})_{ij}$ are i.i.d. and hence they satisfy the Lindeberg condition; that is, summations in \eqref{K28April23.8} go to $0$, as $d$ goes to $\infty$. It follows that the finite sum over $i,j$ in \eqref{K28April23.4} goes to $0$ as well, finishing the proof of our proposition.
\end{proof}

\subsection{Proof of \Cref{prop3.4}} \label{appendix_a2}
\begin{proof}
	Since \Cref{prop3.5} implies that asymptotic normality of $\sigma^{-1} \left(T - \mu \right)$ holds if $\sigma_{\dagger}^2 \rightarrow \infty$ as $d \rightarrow \infty$, all we need for \Cref{prop3.4} is to prove the converse. Let us assume that $\{d_t, t \geq 1\}$ is an increasing sequence of positive integers such that $\lim_{t \rightarrow \infty} \sigma_{d_t}^2 = \sigma_{\dagger}^2 < \infty$. Let 
	\begin{dmath}
		S_d'= \displaystyle \sum_{w=1}^d
		\sum_{l=1}^m 
		\sum_{u, v = 1}^K 
		\sum_{i, j=1}^n    
		\left[ \left((\widetilde{\bm A}^{G_l}_w)_{ij} - (\widetilde{\bm A}^{H_l}_w)_{\pi (\widehat{\bm Q}_{\text{com}})(ij)} \right)^2 1_{\{i, j \in C_u \&
			\tilde{p}_{uu}^l \leq 1/2 \}} \\
		+ \left((\widetilde{\bm A}^{G_l}_w)_{ij} - (\widetilde{\bm A}^{H_l}_w)_{\pi (\widehat{\bm Q}_{\text{com}})(ij)} \right)^2 1_{\{i \in C_u, j \in C_v, u \neq v \&
			\tilde{q}_{uv}^l \leq 1/2 \}} \right] 
	\end{dmath}
	and 
	\begin{dmath}
		S_d'' = \displaystyle \sum_{w=1}^d
		\sum_{l=1}^m 
		\sum_{u, v = 1}^K 
		\sum_{i, j=1}^n    
		\left[ \left((\widetilde{\bm A}^{G_l}_w)_{ij} - (\widetilde{\bm A}^{H_l}_w)_{\pi (\widehat{\bm Q}_{\text{com}})(ij)} \right)^2 1_{\{i, j \in C_u \&
			\tilde{p}_{uu}^l > 1/2 \}} \\
		+ \left((\widetilde{\bm A}^{G_l}_w)_{ij} - (\widetilde{\bm A}^{H_l}_w)_{\pi (\widehat{\bm Q}_{\text{com}})(ij)} \right)^2 1_{\{i \in C_u, j \in C_v, u \neq v \&
			\tilde{q}_{uv}^l > 1/2 \}} \right] 
	\end{dmath}
	Since $E(S_d') + E(S_d'') \leq 2 \left(\text{Var}(S_d') + \text{Var}(S_d'') \right) = \sigma_{d_t}^2 \rightarrow \sigma_{\dagger}^2$ as $t \rightarrow \infty$. It follows that $\max\{\tilde{p}_{uu}^l 1_{\{\tilde{p}_{uu}^l \leq 1/2 \}}, \tilde{q}_{uv}^l 1_{\{\tilde{q}_{uv}^l \leq 1/2 \}}\} \rightarrow 0$ and $\max\{\tilde{p}_{uu}^l 1_{\{\tilde{p}_{uu}^l > 1/2 \}}, \tilde{q}_{uv}^l 1_{\{\tilde{q}_{uv}^l > 1/2 \}}\} \rightarrow 0$ as $d = d_t$ and $t \rightarrow \infty$. By the proof by contradiction, $\sigma_{\dagger}^2 \rightarrow \infty$.
\end{proof}

\subsection{Simulation Results} \label{appendix_a3}
\begin{center}
	\captionof{table}{$P(\text{Type I error})$ for imbalanced block sizes} \label{tab6}
	\centering 	
	\setlength{\tabcolsep}{0.9em}
	\renewcommand{\arraystretch}{1.1}
	
	\smallskip
	\small
	\begin{tabular}{cclccccccc}
		\cline{1-10}
		& \multicolumn{5}{c|}{$\epsilon=0.005$}
		& \multicolumn{4}{c}{$\epsilon=0.01$}
		\\ 
		\hline
		\multicolumn{1}{c|}{$d$} & \multicolumn{2}{c|}{$K=2$}  & \multicolumn{1}{c|}{$K=3$} & \multicolumn{1}{c|}{$K=5$} & \multicolumn{1}{c|}{$K=6$} & \multicolumn{1}{c|}{$K=2$} & \multicolumn{1}{c|}{$K=3$} & \multicolumn{1}{c|}{$K=5$} & \multicolumn{1}{c}{$K=6$} \\ \hline
		\multicolumn{1}{c|}{2}  & \multicolumn{2}{c|}{0.060} & \multicolumn{1}{c|}{0.068}  & \multicolumn{1}{c|}{0.048}  & \multicolumn{1}{c|}{0.059}  & \multicolumn{1}{c|}{0.064}  & \multicolumn{1}{c|}{0.061}  & \multicolumn{1}{c|}{0.060}  &  0.056 \\ 
		\multicolumn{1}{c|}{3}  & \multicolumn{2}{c|}{0.061} & \multicolumn{1}{c|}{0.059}  & \multicolumn{1}{c|}{0.056}  & \multicolumn{1}{c|}{0.060}  & \multicolumn{1}{c|}{0.062}  & \multicolumn{1}{c|}{0.074}  & \multicolumn{1}{c|}{0.058}  &  0.069 \\ 
		\multicolumn{1}{c|}{4}  & \multicolumn{2}{c|}{0.074} & \multicolumn{1}{c|}{0.055}  & \multicolumn{1}{c|}{0.076}  & \multicolumn{1}{c|}{0.066}  & \multicolumn{1}{c|}{0.060}  & \multicolumn{1}{c|}{0.062}  & \multicolumn{1}{c|}{0.065}  &  0.076  \\ 
		\multicolumn{1}{c|}{5}  & \multicolumn{2}{c|}{0.056} & \multicolumn{1}{c|}{0.058}  & \multicolumn{1}{c|}{0.073}  & \multicolumn{1}{c|}{0.075}  & \multicolumn{1}{c|}{0.050}  & \multicolumn{1}{c|}{0.066}  & \multicolumn{1}{c|}{0.069}  &  0.085 \\ 
		\multicolumn{1}{c|}{6}  & \multicolumn{2}{c|}{0.063} & \multicolumn{1}{c|}{0.070}  & \multicolumn{1}{c|}{0.064}  & \multicolumn{1}{c|}{0.067}  & \multicolumn{1}{c|}{0.053}  & \multicolumn{1}{c|}{0.081}  & \multicolumn{1}{c|}{0.062}  &  0.070 \\
		\multicolumn{1}{c|}{7}  & \multicolumn{2}{c|}{0.073} & \multicolumn{1}{c|}{0.065}  & \multicolumn{1}{c|}{0.071}  & \multicolumn{1}{c|}{0.054}  & \multicolumn{1}{c|}{0.080}  & \multicolumn{1}{c|}{0.072}  & \multicolumn{1}{c|}{0.069}  &  0.060 \\ 
		\multicolumn{1}{c|}{8}  & \multicolumn{2}{c|}{0.064} & \multicolumn{1}{c|}{0.062}  & \multicolumn{1}{c|}{0.070}  & \multicolumn{1}{c|}{0.063}  & \multicolumn{1}{c|}{0.063}  & \multicolumn{1}{c|}{0.072}  & \multicolumn{1}{c|}{0.059}  &  0.074 \\ 
		\multicolumn{1}{c|}{9}  & \multicolumn{2}{c|}{0.060} & \multicolumn{1}{c|}{0.054}  & \multicolumn{1}{c|}{0.065}  & \multicolumn{1}{c|}{0.056}  & \multicolumn{1}{c|}{0.054}  & \multicolumn{1}{c|}{0.080}  & \multicolumn{1}{c|}{0.060}  &  0.067 \\ 
		\multicolumn{1}{c|}{10} & \multicolumn{2}{c|}{0.052} & \multicolumn{1}{c|}{0.044}  & \multicolumn{1}{c|}{0.067}  & \multicolumn{1}{c|}{0.073}  & \multicolumn{1}{c|}{0.070}  & \multicolumn{1}{c|}{0.076}  & \multicolumn{1}{c|}{0.063}  &  0.081 \\
		\hline
	\end{tabular}   
\end{center}

\begin{center}
	\captionof{table}{Test power for imbalanced block sizes} \label{tab7}
	\centering 	
	\setlength{\tabcolsep}{0.9em}
	\renewcommand{\arraystretch}{1.1}
	
	\smallskip
	\small
	\begin{tabular}{cclccccccc}
		\cline{1-10}
		& \multicolumn{5}{c|}{$\epsilon=0.005$}
		& \multicolumn{4}{c}{$\epsilon=0.01$}  \\ 
		\hline
		\multicolumn{1}{c|}{$d$}  & \multicolumn{2}{c|}{$K=2$} & \multicolumn{1}{c|}{$K=3$} & \multicolumn{1}{c|}{$K=5$} & \multicolumn{1}{c|}{$K=6$} & \multicolumn{1}{c|}{$K=2$} & \multicolumn{1}{c|}{$K=3$} & \multicolumn{1}{c|}{$K=5$} & \multicolumn{1}{c}{$K=6$} \\ 
		\hline
		\multicolumn{1}{c|}{2}  & \multicolumn{2}{c|}{0.633} & \multicolumn{1}{c|}{0.715}  & \multicolumn{1}{c|}{0.728}  & \multicolumn{1}{c|}{0.696}  & \multicolumn{1}{c|}{0.965}  & \multicolumn{1}{c|}{0.962}  & \multicolumn{1}{c|}{0.978}  &  0.981 \\ 
		\multicolumn{1}{c|}{3}  & \multicolumn{2}{c|}{0.716} & \multicolumn{1}{c|}{0.743}  & \multicolumn{1}{c|}{0.795}  & \multicolumn{1}{c|}{0.771}  & \multicolumn{1}{c|}{0.976}  & \multicolumn{1}{c|}{0.985}  & \multicolumn{1}{c|}{0.993}  &  0.988 \\ 
		\multicolumn{1}{c|}{4}  & \multicolumn{2}{c|}{0.768} & \multicolumn{1}{c|}{0.807}  & \multicolumn{1}{c|}{0.816}  & \multicolumn{1}{c|}{0.802}  & \multicolumn{1}{c|}{0.985}  & \multicolumn{1}{c|}{0.994}  & \multicolumn{1}{c|}{0.998}  &  0.994 \\ 
		\multicolumn{1}{c|}{5}  & \multicolumn{2}{c|}{0.792} & \multicolumn{1}{c|}{0.835}  & \multicolumn{1}{c|}{0.868}  & \multicolumn{1}{c|}{0.833}  & \multicolumn{1}{c|}{0.991}  & \multicolumn{1}{c|}{0.991}  & \multicolumn{1}{c|}{0.992}  &  0.998 \\ 
		\multicolumn{1}{c|}{6}  & \multicolumn{2}{c|}{0.847} & \multicolumn{1}{c|}{0.850}  & \multicolumn{1}{c|}{0.864}  & \multicolumn{1}{c|}{0.863}  & \multicolumn{1}{c|}{0.999}  & \multicolumn{1}{c|}{0.994}  & \multicolumn{1}{c|}{0.997}  &  0.999 \\ 
		\multicolumn{1}{c|}{7}  & \multicolumn{2}{c|}{0.850} & \multicolumn{1}{c|}{0.893}  & \multicolumn{1}{c|}{0.908}  & \multicolumn{1}{c|}{0.876}  & \multicolumn{1}{c|}{0.994}  & \multicolumn{1}{c|}{0.996}  & \multicolumn{1}{c|}{1.000}  &  0.999 \\ 
		\multicolumn{1}{c|}{8}  & \multicolumn{2}{c|}{0.855} & \multicolumn{1}{c|}{0.895}  & \multicolumn{1}{c|}{0.899}  & \multicolumn{1}{c|}{0.890}  & \multicolumn{1}{c|}{0.998}  & \multicolumn{1}{c|}{0.999}  & \multicolumn{1}{c|}{0.998}  &  0.999 \\ 
		\multicolumn{1}{c|}{9}  & \multicolumn{2}{c|}{0.853} & \multicolumn{1}{c|}{0.886}  & \multicolumn{1}{c|}{0.909}  & \multicolumn{1}{c|}{0.898}  & \multicolumn{1}{c|}{0.995}  & \multicolumn{1}{c|}{0.996}  & \multicolumn{1}{c|}{0.999}  &  0.999 \\ 
		\multicolumn{1}{c|}{10} & \multicolumn{2}{c|}{0.897} & \multicolumn{1}{c|}{0.907}  & \multicolumn{1}{c|}{0.926}  & \multicolumn{1}{c|}{0.907}  & \multicolumn{1}{c|}{0.995}  & \multicolumn{1}{c|}{0.999}  & \multicolumn{1}{c|}{1.000}  &  0.998 \\
		\hline
	\end{tabular}   
\end{center}

\subsection{Some Useful Results} \label{appendix_a4}
\subsubsection{Relationship between \texorpdfstring{$\hat{\bm p}$, $\hat{\bm p}^*$, $\hat{\bm q}$, $\hat{\bm q}^*$}{hat p, hat p*, hat q, hat q*} and \texorpdfstring{$\bm p$, $\bm p^*$, $\bm q$, $\bm q^*$}{p, p*, q, q*}} \label{A.4.1}

For $u \textsuperscript{th}$ and $v \textsuperscript{th}$ communities in graph $G$ with $u \neq v$, we have 

\begin{center}
	$\displaystyle \hat{p}_{uu} = \frac{\# \text{ of } 1\text{'s in } \bm A_{uu}}{2{n_u^G \choose 2}}$, and
	$\displaystyle \hat{q}_{uv} = \frac{\# \text{ of } 1\text{'s in } \bm A_{uv}}{{n_u^G \choose 1}{n_v^G \choose 1}}$.
\end{center}
Hence, 
\begin{center}
	$\displaystyle E(\hat{p}_{uu}) = \frac{n_u^G(n_u^G - 1) p_{uu}}{2{n_u^G \choose 2}} = p_{uu}$,
	$\displaystyle E(\hat{q}_{uv}) = \frac{n_u^G n_v^G q_{uv}}{{n_u^G \choose 1}{n_v^G \choose 1}} = q_{uv}$.\\
\end{center}
Similarly, 
\begin{center}
	$\displaystyle E(\hat{p}_{uu}^*) = \frac{n_u^H(n_u^H - 1) p_{uu}^*}{2{n_u^H \choose 2}} = p_{uu}^*$,
	$\displaystyle E(\hat{q}_{uv}^*) = \frac{n_u^H n_v^H q_{uv}^*}{{n_u^H \choose 1}{n_v^H \choose 1}} = q_{uv}^*$. \\
\end{center}
Hence, $\hat{\bm p}$, $\hat{\bm p}^*$, $\hat{\bm q}$, $\hat{\bm q}^*$ are unbiased estimators of $\bm p$, $\bm p^*$, $\bm q$, $\bm q^*$, respectively.

\subsubsection{Relationship between $\sigma^2$ and $n$, $K$} \label{A.4.2}
$C1.$ If $i$ and $j$ are in $C_u$, then $\sigma_{ij}^2 = \hat{p}_{uu}(1 - \hat{p}_{uu})$.
\begin{enumerate}[label=(\roman*)]
	\item Consider the balanced community size, 
	\begin{center}
		$\displaystyle \hat{p}_{uu} = \frac{\sum_{i, j \in C_u} \bm A_{ij}}{\frac{n}{K}(\frac{n}{K} - 1)}$.
	\end{center}
	
	\item Consider the imbalanced community size, 
	\begin{center}
		$\displaystyle \hat{p}_{uu} = \frac{\sum_{i, j \in C_u} \bm A_{ij}}{\frac{n \exp{(w_u)}}{\sum_{u=1}^K \exp{(w_u)}}\left(\frac{n \exp{(w_u)}}{ \sum_{u=1}^K \exp{(w_u)}} - 1 \right)}$.
	\end{center}
\end{enumerate} \leavevmode
\\
$C2.$ If $i$ is in $C_u$ and $j$ is in $C_v$ with $u \neq v$, then $\sigma_{ij}^2 = \hat{q}_{uv}(1 - \hat{q}_{uv})$.
\begin{enumerate}[label=(\roman*)]
	\item Consider the balanced community size, 
	\begin{center}
		$\displaystyle \hat{q}_{uv} = \frac{\sum_{i \in C_u, j \in C_v} \bm A_{ij}}{\frac{n^2}{K^2}}$.
	\end{center}
	
	\item Consider the imbalanced community size, 
	\begin{center}
		$\displaystyle \hat{q}_{uv} = \frac{\sum_{i \in C_u, j \in C_v} \bm A_{ij}}{\frac{n^2 \exp{(w_u)} \exp{(w_v)}}{\left (\sum_{u=1}^K \exp{(w_u)} \right)^2}}$.
	\end{center}
\end{enumerate}

\medskip
%Bibliographystyle{IEEEtran}
%Bibliography{reference}{}
%Bibliographystyle{plainnat}
\bibliographystyle{apacite}	
\bibliography{reference}
	
\end{document}